\documentclass[11pt]{article}

\def\arXiv{1}

\ifx\comments\undefined
  \usepackage[letterpaper,margin=1in,includefoot]{geometry}
  \newcommand{\XXXcomment}[1]{}
  \newcommand{\XXXcommentR}[1]{}
  \newcommand{\XXXcommentG}[1]{}
\else
  \usepackage[letterpaper,margin=.5in,right=2.5in,marginpar=2in,includefoot]{geometry}
  \newcommand{\XXXcomment}[1]{\marginpar{\color{blue}{\footnotesize #1}}}
  \newcommand{\XXXcommentR}[1]{\marginpar{\color{red}{\footnotesize #1}}}
  \newcommand{\XXXcommentG}[1]{\marginpar{\color{green}{\footnotesize #1}}}
\fi

\ifx\arXiv\undefined
  \usepackage[pdftex,
              bookmarks=true,
              bookmarksnumbered=true,
              pdfborder={0 0 0},
              plainpages=false,
              unicode,
              pdfpagelabels]{hyperref}
\fi

\usepackage{amssymb,amsthm,latexsym,amsmath,mathrsfs,stmaryrd,graphicx,fancyhdr,paralist,longtable}

\usepackage{bussproofs}

\usepackage{linguex}

\ifx\arXiv\undefined
  \usepackage[sort]{natbib}
\fi

\usepackage{float}
\floatstyle{boxed}
\restylefloat{figure}
\restylefloat{table}

\usepackage{caption,float}
\usepackage{tikz}
\usetikzlibrary{trees,arrows,automata,shapes,decorations,decorations.pathmorphing}

\newcommand{\mystretch}{\renewcommand{\arraystretch}{1.5}}
\newcommand{\normalstretch}{\renewcommand{\arraystretch}{1}}

\newcommand{\imp}{\to}                         

\newcommand{\may}[1]{\langle{#1}\rangle}       

\newcommand{\sem}[1]{\llbracket{#1}\rrbracket} 

\newcommand{\col}{\,{:}\,}                     

\newcommand{\RBB}{{\mathsf{RBB}}}              
\newcommand{\QRBB}{{\mathsf{QRBB}}}            

\newcommand{\JL}{{\mathsf{JL}}}                
\newcommand{\JTB}{{\mathsf{JTB}}}              
\newcommand{\NIL}{{\mathsf{NIL}}}              

\theoremstyle{definition} 
\newtheorem{definition}{Definition}[section]
\newtheorem{theorem}[definition]{Theorem}

\newtheorem{lemma}[definition]{Lemma}
\newtheorem{corollary}[definition]{Corollary}

\newcommand{\ourtitle}{Knowledge, Justification, and Adequate Reasons}

\newcommand{\paulENS}{Paul \'Egr\'e}

\newcommand{\paulENSinstitute}{CNRS / ENS / EHESS\\ \small{PSL University}}

\newcommand{\paulMIT}{Paul Marty}

\newcommand{\paulMITinstitute}{ZAS Berlin}

\newcommand{\bryan}{Bryan Renne}

\newcommand{\bryanInstitute}{University of Saskatchewan}

\newcommand{\bryanFunding}{Bryan Renne was funded by an Innovational
  Research Incentives Scheme Veni grant from the Netherlands
  Organisation for Scientific Research (NWO) and was employed by ILLC Amsterdam at the start of this project.}
  
\newcommand{\paulENSFunding}{This is a substantially revised and expanded version of our earlier paper titled ``Knowledge, Justification and Reason-Based Belief'', a short version of which appeared in the \emph{Proceedings of the Amsterdam Colloquium 2015}, ed. by T. Brochhagen, F. Roelofsen and N. Theiler (pp. 100-108).
Paul \'Egr{\'e} was funded 
  by the ANR
  project ``Trivalent Logics and Natural Language Meaning'' (ANR-14-CE30-0010) and by the ANR program FrontCog ANR-17-EURE-0017 for research conducted at the Department of Cognitive Studies of ENS.}
  
  \newcommand{\martyFunding}{Paul Marty was funded by the MIT Linguistics Department when this project started.}

\ifx\arXiv\undefined
\hypersetup{ 
  pdfauthor={\paulENS{}, \paulMIT{}, \bryan{}}, 
  pdftitle={\ourtitle{}}
}
\fi

\title{\ourtitle{}\thanks{\paulENSFunding{} \bryanFunding{} \martyFunding}}
\renewcommand\footnotemark{}

\author{
  \paulENS{}\\
  \small\paulENSinstitute{}
  \and
  \paulMIT{}\\
  \small\paulMITinstitute{}
  \and
  \bryan{}\\
  \small\bryanInstitute{}
}

\date{}

\begin{document}
\maketitle
\thispagestyle{fancy}

\begin{abstract}
  \noindent Is knowledge definable as justified true belief (``JTB'')?
  We argue that one can legitimately answer positively or negatively,
  depending on whether or not one's true belief is justified by what
  we call \emph{adequate} reasons.  To facilitate our argument we
  introduce a simple propositional logic of \emph{reason-based belief}, and
  give an axiomatic characterization of the notion of adequacy for
  reasons.  We show that this logic is sufficiently flexible to
  accommodate various useful features, including quantification over
  reasons.  We use our framework to contrast two notions of JTB: one
  internalist, the other externalist. We argue that Gettier cases
  essentially challenge the internalist notion but not the externalist
  one. Our approach commits us to a form of infallibilism about
  knowledge, but it also leaves us with a puzzle, namely whether
  knowledge involves the possession of only adequate reasons, or
  leaves room for some inadequate reasons. We favor the latter
  position, which reflects a milder and more realistic version of
  infallibilism.
  \end{abstract}


\section{Introduction}

Can the ordinary concept of knowledge be defined in terms of justified
true belief (``JTB'')? Since Gettier's paper \cite{Get63:A}, the
answer to this question is widely considered negative. Gettier
produced two cases intended to show that a belief can be true,
justified, and yet fall short of knowledge. The first concerns Smith,
an applicant for a job who has ``strong evidence'' that Jones is the
man who will get the job and also that Jones has ten coins in his
pocket. Unknown to Smith, it turns out that Smith himself has ten
coins in pocket and is actually the one selected for the job. Smith's
belief that ``the man who will get the job has ten coins in his
pocket'' is therefore true and justified, but it seems inappropriate
to say that this belief constitutes knowledge. The second case is one
in which Smith believes a false proposition $p$ based on persuasive
evidence for $p$ and infers from $p$ some true proposition $p \lor q$
by picking the true disjunct $q$ at random. Here too, Smith
justifiably believes $p \lor q$, but it seems incorrect to say that he
knows $p\lor q$.

 Our point of departure in this paper is the following: even though we
agree with the force of Gettier's examples, we share with others (in particular \cite{Dre71:CR,Sos74:how, Sos79:presupp,Chis77:TOK,Gol79:WJB,Tur12:S}) the intuition that those examples do not necessarily invalidate every analysis of knowledge in terms of justified true belief, depending on how the notion of justification is understood. 
Indeed, what
Gettier's examples show is that an agent can have an
internal justification for believing a proposition that is \emph{plausible}, without that
justification being properly \emph{adequate} to the truth of the
proposition in question. But if so, Gettier cases must only show that
knowledge is not identical with JTB under an internalist conception of
justification. The examples do not thereby rule out the existence of a
more externalist notion of justification capable of sustaining the
equation between knowledge and JTB. Define knowledge as true belief
with an \emph{adequate} justification, and it appears Gettier cases no
longer have a bite. Our leading intuition in that regard is shared with the earlier analyses of knowledge by Chisholm in \cite{Chis77:TOK} and Sosa in \cite{Sos74:how, Sos79:presupp}, who define knowledge as the possession of a \emph{non-defective} justification, and by Dretske in \cite{Dre71:CR}, who defines knowledge as the possession of a \emph{conclusive} reason.

Admittedly, a definition of knowledge along those lines might not provide a noncircular or
reductive analysis (see \cite{Wil00:Book}): notions of ``non-defectiveness'', ``conclusiveness'', or indeed ``adequacy''
may ultimately have to be understood in ways that presuppose a prior
grasp of the concept of knowledge. For example, if an adequate
justification were to mean ``a justification that is suitable to make
the belief count as knowledge,'' then it would appear that we define
knowledge in terms of itself. We agree with this objection, but the
notion of adequacy may also turn out to not wholly depend on epistemic
notions. Adequacy, for instance, may be at least partly characterizable in terms
of truth-making, and the truth-making relation need not refer to prior
epistemic notions. Or consider the relation between a fully formalized
axiomatic proof and a mathematical statement derived in that proof:
the ``adequacy'' of the proof as a vehicle for mathematical truth is a
purely syntactic notion, with no epistemological concepts
presupposed. These examples suggest to us that room remains for a
fruitful investigation of the concepts of knowledge, belief, and
justification that acknowledges the distinction between adequate and
inadequate reasons.

The gist of our account lies in the distinction
between reasons that \emph{(merely) support} belief in a proposition
and reasons that are not only supportive but are also what we call
\emph{adequate}. Our first goal in this paper is to give an axiomatic characterization of both concepts, and to use them to clarify the duality found in the concept of justified true belief. Our characterization treats adequate reasons basically as externalist noninferential justifications, following Fumerton's typology in \cite{Fum02:TOJ}. They are noninferential in the sense that we do not require agents to be able to justify their plausibility by a further reason, and externalist in that adequacy is not necessarily a property an agent can ascertain. Like Chisholm and Sosa in their respective treatments of nondefective justification, or Dretske in his treatment of conclusive reasons, we moreover treat adequate reasons as being infallible, in the sense that they can only support true propositions. The latter property is not a definition of adequacy, however, but only a central property, as we shall explain.

A second goal we have on that basis is more technical: it consists, in
the wake of work done in Justification Logic
\cite{Art08:RSL,ArtFit11:SEP}, in giving an explicit treatment of
reasons in epistemic logics but with specific emphasis on the notion
of an adequate reason.\footnote{\cite{Art08:RSL,ArtFit11:SEP} uses a
  formal framework to track what goes wrong with specific lines of
  reasoning in the examples of Gettier, Goldman, and Kripke; and
  \cite{BalRenSme12:LNCS,BalRenSme14:APAL} uses a related framework
  that has additional features from Belief Revision Theory to reason
  about the examples of Lehrer and Gettier. Our work here is
  different.  First, while inspired by the Justification Logic
  approach to reasoning about justifications, our setting is in many
  respects simplified but at the same time includes certain novel
  features (see \S\ref{section:closure} for details). Second, our task
  here is different than that in \cite{Art08:RSL,ArtFit11:SEP}: beyond
  providing a formal diagnosis of what goes wrong in certain
  Gettier-type examples, we define different kinds of JTB and discuss
  their relative susceptibility to such examples. This has some
  similarities with the work in
  \cite{BalRenSme12:LNCS,BalRenSme14:APAL}; however, our logics are
  much simpler and we use them as part of a different analysis.}  Our
third goal is more philosophical. As already suggested, our account
commits us to a version of infallibilism about knowledge, for an
adequate reason supports only true propositions. However, our beliefs
for a proposition often are grounded in two kinds of reasons, some of
them adequate, and others inadequate. The question we are interested
in concerns whether knowledge should be defined in terms of the
possession of only adequate reasons, or whether it tolerates the
inclusion of inadequate reasons. This is what we call \emph{the
  problem of mixed reasons}. We address this problem and defend the
view that knowledge should be made compatible with the possession of
mixed reasons.

In order to articulate the distinction between two kinds of justification, in \S\ref{section:rsa} we first present
the basic concepts of our account of knowledge and justification, namely the concepts of reason, support, and of adequacy. We then introduce the \emph{Logic of Reason-Based Belief} in \S\ref{section:logics}.
This logic provides an explicit representation of reasons to believe a
proposition.  We show that the logic is sufficiently flexible to
accommodate quantifiers over reasons, limited or full closure of
reason-based belief under implication, and an optional requirement
that all beliefs be reason-based.  In \S\ref{section:jtb}, we put this
logic to work in the analysis of Gettier cases: first to tease apart
two notions of justification, one internalist and the other
externalist; and second to study the susceptibility of internal JTB
and of external JTB to Gettier-type examples. Finally, in \S\ref{section:mixed} we close the paper with
a discussion of the problem of mixed reasons and with the discussion of an objection to our account, regarding whether adequacy is even a necessary condition for the ascription of knowledge. Technical notions and results that are not required
for an understanding of the main text are relegated to an appendix.

\section{Reasons, Support, and Adequacy}\label{section:rsa}

We distinguish between two kinds of reasons: those that merely
``support'' a proposition by inclining an agent to believe in it, and
those that not only support a proposition but are also themselves
``adequate.'' Our analysis of Gettier cases hinges on a precise
understanding of this key distinction. In this section we first introduce the three basic concepts that we rely on in our epistemology beside the concept of belief, namely reasons, support, and adequacy. Regarding belief, suffice it to say that we treat belief as an explicit endorsement by an agent of either a proposition, or a reason providing evidence for a proposition. We say more about our characterization of belief in \S\ref{section:logics}, where we further clarify the relation between beliefs and reasons.

\subsection{Reasons} 

The first concept we need to clarify is the concept of a
\emph{reason}. As emphasized by Armstrong in \cite{Arm1973:belief},
``talk of a man's reason may be to refer to a (certain sort of)
belief-state of his, or it may be to refer to the proposition
believed" (pg. 79). We handle reasons primarily as state-specific
objects, and only derivatively as propositions. For us, a reason is
some evidence on the basis of which we come to believe a
proposition. A reason therefore does not have the type of a
proposition, although it can be associated to a proposition in a
systematic manner.\footnote{We are indebted to an anonymous reviewer
  and to J.~Dutant for urging us to clarify that matter. In a previous
  version of our work, we left open the possibility for reason symbols
  to directly denote propositions. Note that throughout the paper, and
  unless when it creates confusion, we use the word ``proposition''
  both for syntactic objects (closed sentences of our language), and
  for the semantic objects they denote (viewed as sets of possible
  worlds).} For example, my hearing voices outside might be a reason
to believe that there are people outside. Semantically, we shall
represent reasons by accessibility relations, thereby treating them as
akin to belief-states. Given a reason $r$ and a world $w$, we write
$r(w)$ to represent the proposition determined by the reason $r$ in
$w$ (that is the set of accessible worlds in virtue of the relation
corresponding to $r$). So, each reason can be associated with a
proposition at a world, but reasons \emph{per se} are not
propositions.\footnote{Note also that given a sentence $p$, $p$ is
  supposed to express the same proposition relative to every world. In
  our framework a reason $r$ may express different propositions $r(w)$
  depending on the world $w$.} Reasons are related to the sort of
answers produced by rational agents when they are asked, ``Why do you
believe this?'' or ``How do you know that?''. If I am asked why I
believe that Napoleon I is dead by now, I would answer that it is
because I read that he was born in the eighteenth century, and he died
in the nineteenth century. If I am asked why I believe that $2+2=4$, I
would respond that my calculations confirm that, or that I accept a
certain axiomatic system that I see as consistently deriving that
fact. So reasons for us are not bare propositions, like the
proposition ``Napoleon died in the 19th century'' or the proposition
``$2+1=3$'', or even bare arguments, but we take them to be related to evidential experience of some sort.

\subsection{Support} 

The second concept we deal with is what we call \emph{support} between
a reason and a proposition. Syntactically, we treat the relation of
support between a reason and a proposition as primitive: we will write
``$r\col \varphi$''. Model-theoretically, we represent the support
relation between a reason $r$ and a proposition $p$ by the fact that
$r(w)$ entails $p$. This does not mean that we intend the relation of
support to model deductive support. Rather, we do intend it to capture
a very fundamental form of inductive support. That is, we view each
reason as a basic Humean experience of some kind, creating an
inclination to believe a proposition. If I hear voices outside the
house, that is a supporting reason for the proposition that there are
people outside the house. In principle, that might also be a reason to
think there are loudspeakers broadcasting voices outside the house
(with no actual person around). In our approach, however, we do not
allow for reasons to support contradictory propositions. If $r$ is a
supporting reason for $\varphi$, then we will take $r$ to be a
supporting reason for any proposition entailed by
$\varphi$. Consequently, a reason for thinking that there are
loudspeakers broadcasting voices outside the house with no one around
will have to be different from a reason to think that there are actual
people outside the house.

A delicate issue is whether the support relation ought to be treated
as subjective (internal to an agent's beliefs) or as objective
(external to the agent's beliefs). That is, could an agent believe
that $r$ supports $\varphi$ without $r$ actually supporting $\varphi$?
Conversely, could $r$ support $\varphi$ without the agent believing
that $r$ supports $\varphi$? Regarding the latter, in our system any
reason supports every logical truth, but we see it as possible for an
agent not to believe that, in particular because an agent may not
automatically see the truth of every logical truth or be aware of
every reason. Hence, a reason may support a proposition without an
agent believing that.

The question whether an agent may believe a reason to support a
proposition without that reason supporting the proposition is more
delicate. One way in which this could happen is if reasons can be
treated as mere propositions. For example, an agent who accepts a
fallacious mathematical argument may believe that true premises
\emph{deductively support} a false conclusion. In that case one may
say that the agent incorrectly believes the premises to deductively
support the conclusion. Again we find it better not to let
propositions be reasons \emph{per se}, but to keep to reasons being
evidence of a sort, related to some experience. Take an agent who
miscalculates the result of some mathematical equation. He accepts the
premise that a bat and a ball cost \$1.10, that the bat costs \$1 more
than the ball, and wrongly computes from that that the ball costs
\$0.10 (it actually costs \$0.05).\footnote{This example is from S.~Frederick and D.~Kahneman (\cite{kahneman2002representativeness}). A variant is discussed
by Sorensen in \cite{Sor15:fugu}.} On an objective reading of reasons, letting
$p_1$ and $p_2$ stand for the two premises of the argument, and $q$
for the false conclusion that the ball costs \$0.10, we would write:
$B((p_{1}\wedge p_{2})\col q)$ to express that the agent believes the
premises to deductively support the conclusion. But we find more
appropriate to write instead that
$B(r:(p_{1} \wedge p_{2} \rightarrow q))$, where $r$ stands for the
calculation made by the agent, to represent the fact that the agent
sees his calculations to support a certain deductive relation between
premises and conclusions. As a result, we treat the relation of
deductive support between bare propositions by means of the usual
logical resources available in our system, and not in terms of the
colon operator.

Because we see reasons as evidence-based propositions, we might wish
to endorse the success principle
$B(r\col \varphi) \rightarrow r\col \varphi$. That is, the agent
believing some kind of experience to support a proposition would
suffice to make the experience in question a supporting reason for
$\varphi$. A rational agent may therefore not necessarily be aware of
all his or her reasons, but could not be mistaken about reasons that
he or she explicitly endorses. One may object that even rational
agents may be deluded about their reasons. For example, an agent might
be mistaken about her own experiences. She has a memory of an
encounter with a famous actor, but that encounter never happened, and
the memory is no real experience. When asked to justify why she
believes this actor is blond, she gives as a reason that she
remembered seeing him to be blond during that encounter. Arguably, the
reason does not support the proposition in that case, because the
reason is not grounded in any true experience. Such cases invite us to
guard against the aforementioned success principle. Because of that,
we will propose a weakening of that principle, intended to secure more
generality, and to not rule out such cases.\medskip

\subsection{Adequacy} 

What about \emph{adequacy}? Adequacy is by far the most central
concept in our approach. There are various ways in which the notion
can be thought of. One option is to think of an adequate reason as a
\emph{reliable} reason (see \cite{Gol79:WJB}). On that approach, an adequate
reason is a reason that produces true beliefs most of the time. Like
other critics of reliabilism, we do not view this characterization as
strong enough. A stronger characterization might be in terms of
Dretske's notion of \emph{conclusive} reason. According to Dretske, $r$ is a
conclusive reason for $p$ if and only if $r$ would not be the case if
$p$ were not the case (\cite{Dre71:CR}). Another approach is in terms
of Chisholm's or Sosa's respective conceptions of \emph{non-defectiveness}
\cite{Chis77:TOK,Sos74:how}, whereby a justification is non-defective
provided it is not the basis of any false proposition. Consistently with those views, our use of the term ``adequate'' is also related to Spinoza's understanding in \emph{Ethics}, where the notions of adequate idea (\emph{adaequata idea}) and adequate knowledge (\emph{adaequata cognitio}) are meant to imply truth.\footnote{See \cite{spinoza1677ethica} Part II, in particular Def. 4, Prop. 34 and Prop. 41.}

We draw inspiration from the latter set of approaches but an important
caveat about our approach is that we do not propose to give an
explicit definition of adequacy in terms of necessary and sufficient
conditions. Instead, we propose to characterize adequacy in terms of
two necessary conditions, and in terms of the interaction with the
notions of support and belief we have in our ontology. The first
central condition we impose on a reason for it to be adequate is very
close to Chisholm's, Sosa's, or indeed Spinoza's respective conceptions. We say that a
reason is adequate only if the propositions it supports are true
(axiom (A) below). But we also see adequacy as putting a requirement
on the support relation. In particular, we consider that if an agent
wrongly believes a reason $r$ to support a proposition $\varphi$ (that
is, believes so although $r$ actually fails to support $\varphi$),
then $r$ cannot be adequate (axiom (AS) below). So, an adequate reason
is such that it secures not only the propositions that it actually
supports, but also such that it secures the correctness of the belief
in the associated support relation. This second requirement basically
corresponds to the weakening of the success condition evoked before,
that $B(r\col \varphi)$ should always entail $r\col \varphi$. It says
that this is so provided $r$ is adequate.

Let us make four important remarks on our treatment of
adequacy:
\begin{enumerate}
\item Our central requirement on adequacy ensures that if $\varphi$ is
  not true, and $r$ supports $\varphi$, then $r$ cannot be
  adequate. This property of adequacy therefore comes close to
  Dretske's, except that we rule out the possibility that a reason
  might be adequate evidence for one proposition and inadequate for
  another (we discuss that aspect in \S\ref{section:mixed}). That
  first requirement implies that reasons that do not validly support a
  given conclusion are inadequate.

\item In our system $B(r\col \varphi)$ does not entail that $Br$,
  which says that an agent does not automatically deem adequate any
  reason she has in support of a proposition. For example, in the
  M\"uller-Lyer illusion, an agent who trusts her perceptual
  experience that two lines differ in length has a reason in support
  of the proposition that the lines are not equal. But she can have a
  different and overriding reason, her measuring of the lines (reason
  $s$), telling her that the lines do not differ in length. We would
  represent the case as:
  \[
  B(r\col p) \wedge B(s\col \neg p) \wedge Bs \wedge \neg Br\enspace.
  \]

\item We do not identify the adequacy of a reason with the validity or
  even the soundness and validity of an argument. Consider the
  following example of what Sorensen \cite{Sor15:fugu} calls
  ``paralemmatic reasoning''. Our agent has 10 cents in his pocket
  (premise 0). He infers from that, and from the information that the
  bat and ball cost \$1.10 (premise 1) and the bat costs \$1 more than
  the ball (premise 2), that he has enough money to buy the ball
  (conclusion). Let us assume he infers the conclusion really because
  he infers that the ball costs 10 cents. Let us represent by
  $p_{0}, p_{1}, p_{2}$ the premises of that argument, by $q$ the
  correct conclusion that the agent has enough money to buy the ball,
  and by $p$ the incorrect lemma that the ball costs 10 cents. Let us
  call $r$ the agent's calculations. The agent believes that
  $r\col (p_{0} \land p_{1} \land p_{2} \rightarrow q)$. The argument
  $p_{0} \land p_{1} \land p_{2} \rightarrow q$ has sound premises,
  and the conclusion validly follows. In effect, however, what $r$
  supports in the agent's mind is a more specific argument, namely the
  argument that $p_{1} \land p_{2} \rightarrow p$, combined with the
  argument that $p \wedge p_{0}\rightarrow q$. Because of our adequacy
  constraint, $r$ cannot be an adequate reason here, because $r$
  supports an incorrect argument. Sorensen in \cite{Sor15:fugu} writes
  that ``paralemmas are precise counterexamples to bare evidentialism
  -- the view that we know a proposition simply by virtue of having
  adequate evidence for it.'' By adequate evidence we take Sorensen to
  refer to some \emph{intrinsic} property of the evidence, holding
  irrespective of the agent's beliefs. This is not our
  characterization of adequacy, precisely because for us a reason can
  support a logically perfect argument without being adequate. We are
  not bare evidentialists as a consequence.

\item Finally, we do not view the two conditions we impose on
  adequacy, namely (A) and (AS), as sufficient conditions for a reason
  being adequate, but only as necessary conditions. For us, and as our
  semantics will show, adequacy is fundamentally a property of a
  reason and a world. Or put otherwise, the adequacy of a reason
  represents the property of a world being a \emph{good case} for the
  reason in question (see \cite{Wil00:Book}). The same reason, with
  the same intrinsic properties, could turn out to be adequate in a
  good case and inadequate in a different case that is bad. Being a
  good case is not a notion we think we can explicitly define and
  exhaustively characterize. Intuitively, being a good case (relative
  to a reason) means being a case at which the reason is used in a way
  that is not lucky, that is safe, or that has the right fit, whatever
  those notions might mean. We consider our two axioms on adequacy to
  capture part of that intuition, but some of the cases we discuss
  (particularly Goldman cases, see below) are arguably cases in which
  both of our conditions are fulfilled, but which we do not want to
  characterize as adequate in the externalist sense we think is
  relevant.
\end{enumerate}

Our axiomatic treatment of adequacy therefore implies that the
adequacy of a reason does not lead to any false conclusion, and it
also implies that the agent is not deluded about the experience
reported in the reason. But to rehearse our main point, those are only
two necessary conditions on adequacy. We view adequacy as an even
stronger property of a reason and a world, namely as the property for
a world to secure that the reason is not merely internally but also
externally warranted.

\section{A Simple Logic of Reason-Based Belief}
\label{section:logics}

Having laid out the basic ingredients in our ontology, in this section
we present a simple logic of reason-based belief. We first present our
system axiomatically and then give its underlying model theory. We
then present various extensions of the basic system, in particular to
allow for quantification over reasons, which will be needed in our
analysis of justified true belief in the next section.

\subsection{Reason-Based Belief}
\label{section:RBB}

Fix nonempty sets $P$ of propositional letters and $R$ of reason
symbols (also called ``reasons'').  $F$ is the set of formulas
$\varphi$ defined by the following grammar:
\begin{eqnarray*}
  \varphi &::=&
  p \mid \lnot\varphi \mid (\varphi\lor\varphi) \mid
  (r\col\varphi) \mid r \mid B\varphi
  \\
  &&
  \text{\footnotesize
    $p\in P$, $r\in R$}
\end{eqnarray*}
Other Boolean connectives are defined as abbreviations.
\begin{itemize}
\item $r\col\varphi$ says, ``$r$ supports $\varphi$''.

\item $r$ says, ``$r$ is an adequate reason.'' 

  For the sake of clarity, we might have introduced a unary operator
  $A$ in the language, and written $Ar$ to express ``$r$ is an
  adequate reason''. We choose to write $r$ instead of $Ar$ to save
  symbols. Our logic will guarantee that every proposition supported
  by an adequate reason is true (i.e., we will have
  $r\col\varphi\to(r\to\varphi)$ for all formulas $\varphi$). Note
  that, as we define it, adequacy is not relativized to a specific
  proposition. This makes our treatment of adequacy different from
  Dretske's treatment of conclusiveness. For Dretske, a reason is not
  conclusive \emph{per se} but conclusive \emph{for} a proposition
  (\cite{Dre71:CR}). In our setting, when a reason is adequate, it
  makes every proposition it supports true. Hence no reason can be
  such that it is adequate relative to one proposition it supports,
  and inadequate relative to another. This feature of our account may
  eventually prove too strong (we return to this issue in the last
  section of the paper), but our proposal is to think of adequacy as a
  fundamental property of truth-conduciveness.

\item $B\varphi$ says, ``the agent believes $\varphi$.''

  The formula $Br$ is therefore read, ``the agent believes $r$ is an
  adequate reason.'' Sometimes it will be convenient to say that ``the
  agent accepts reason $r$'' to mean that $Br$ holds. So believing
  reason $r$ is adequate and accepting reason $r$ will be considered
  to mean the same thing.
\end{itemize}
To reduce the number of parentheses while ensuring unique readability
of formulas, we adopt the convention that the colon operator binds
more strongly than any Boolean connective.  For example,
\begin{eqnarray*}
  r\col\varphi\to\psi
  &\text{denotes}&
  (r\col\varphi)\to\psi
  \quad
  \bigl[\text{and not }
  r\col(\varphi\to\psi) \bigr]\enspace.
\end{eqnarray*}

The theory $\RBB$ of \emph{Reason-Based Belief} appears in
Table~\ref{table:RBB}. We write $\vdash\varphi$ to mean that $\varphi$
is derivable in $\RBB$.  The negation is written $\nvdash\varphi$.

\begin{table}[ht]
  \begin{center}
    \textsc{Axiom Schemes}\\[.4em]
    \renewcommand{\arraystretch}{1.3}
    \begin{tabular}[t]{cl}
      (CL) &
             Axiom Schemes of Classical Propositional Logic
      \\
      (RK) &
             $r\col(\varphi\to\psi)\to
             (r\col\varphi\to r\col\psi)$
      \\
      (A) &
            $r\col\varphi\to(r\to\varphi)$
      \\
      (BRK) & 
              $B(r\col (\varphi\to\psi)) \to
              (B(r\col \varphi)\to B(r\col\psi))$
      \\
      (BA) & 
             $B(r \col \varphi) \rightarrow (Br \rightarrow B\varphi)$
      \\
      (AS) & 
             $B(r\col \varphi) \to (r \rightarrow (r\col \varphi))$
      \\
      (D) &
      $B\varphi\to\lnot B\lnot\varphi$
      
    \end{tabular}
    \renewcommand{\arraystretch}{1.0}
    \\[1em]
    \textsc{Rules}\vspace{-.5em}
    \[
    \frac{\varphi\imp\psi \quad \varphi}{\psi}
    \;\text{\footnotesize(MP)}
    \qquad
    \frac{\varphi}{r\col\varphi}
    \;\text{\footnotesize(RN)}
    \qquad
    \frac{\varphi\leftrightarrow\psi}{B\varphi\leftrightarrow B\psi}
    \;\text{\footnotesize(E)}
    \]
  \end{center}
  \caption{The theory $\RBB$}
  \label{table:RBB}
\end{table}

Regarding the axioms and rules of $\RBB$ from Table~\ref{table:RBB},
(CL) and (MP) say that $\RBB$ is an extension of classical
propositional logic. (RK) says that reasons are closed under material
implication, and (RN) says that reasons support all derivable
formulas. (A) says that if $r$ supports $\varphi$ and $r$ is an
adequate reason, then $\varphi$ is true.  (BRK) says that if the agent
believes $r$ supports a conditional, and believes $r$ supports the
antecedent, then the agent believes $r$ supports the consequent.  (BA)
says that if the agent believes $r$ supports $\varphi$ and the agent
believes that $r$ is an adequate reason, then the agent believes
$\varphi$. (AS) says that if the agent believes $r$ to support a
proposition, then $r$ does not support the proposition unless $r$ is
adequate. (D) says that the agent's beliefs are consistent: she cannot
have contradictory beliefs (i.e., believe both $\varphi$ and
$\lnot\varphi$ for some $\varphi$). Finally, (E) says that the agent's
beliefs do not distinguish between provably equivalent formulas.

As for mnemonics, (CL) is ``Classical Logic,'' (MP) is ``Modus
Ponens,'' (RK) is Kripke's axiom $\mathsf{K}$ of modal logic (used
here for reasons), (RN) is ``Reason Necessitation,'' (A) is
``Adequacy,'' (BRK) is ``the Belief version of RK,'' (BA) is ``the
Belief version of adequacy,'' (AS) is ``Adequate Support,'' (D) is a
well-known axiom from modal logic \cite{Chellas:ml}, and (E) is a
well-known rule from minimal modal logic \cite{Chellas:ml}.

\subsection{Semantics for $\RBB$}
\label{section:RBB-semantics}

We now present a semantics for our system. The semantics is based on
two main components: a neighborhood semantics for belief, intended to
make belief as weak as possible, and a Kripke semantics for reasons,
this time intended to capture the closure properties on
reasons.\footnote{For other systems dealing with belief in terms of a
  neighborhood semantics, see \cite{BenPac11:evidence} and
  \cite{Dut10:PHD,Dut15:Method}.} We justify both desiderata
in the next section.

We construct the models of $\RBB$ in a couple of stages. We first
begin with \emph{pre-models}, which are structures $M=(W,[\cdot],N,V)$
having:
\begin{itemize}
\item a nonempty set $W$ of ``possible worlds,''

\item a function $[\cdot]:R\to\wp(W\times W)$ mapping each reason
  $r\in R$ to a binary relation $[r]\subseteq W\times W$ on the set of
  possible worlds,

\item a ``neighborhood function'' $N:W\to\wp(\wp(W))$ mapping each
  world $w$ to a collection $N(w)$ of sets of worlds
  (``propositions'') that the agent believes at $w$,

\item a propositional valuation $V:W\to\wp(P)$ mapping each world $w$
  to the set $V(w)$ of propositional letters that are true at $w$.
\end{itemize}
To indicate that the components $W$, $[\cdot]$, $N$, and $V$ come from
pre-model $M$, we may write $W^M$, $[\cdot]^M$, $N^M$, and $V^M$
(respectively).  Letting $w\in W$ be a world, $r\in R$ be a reason,
and $X\subseteq W$ be a set of reasons (i.e., a ``proposition''), we
introduce the following abbreviations:
\begin{itemize}
\item $r(w) := \{v\in W\mid (w,v)\in[r]\}$ is the set of all worlds
  that are $r$-accessible from $w$ (i.e., accessible by the relation
  $[r]$).

\item $r^\circ := \{ v\in W \mid (v,v)\in[r]\}$ is the set of all
  worlds that are $r$-accessible from themselves.  These are the
  worlds at which the reason $r$ is said to be \emph{reflexive}. As we
  will see shortly, a reason that is reflexive at a world will be
  adequate at that world. So reflexivity and adequacy are equivalent
  notions, and hence we may conflate the two, which ought not cause
  confusion.
\end{itemize}
To indicate the sets $r(w)$ and $r^\circ$ arise from worlds in
pre-model $M$, we may write $r^M(w)$ and $r^{M\circ}$ (respectively).
Though not mentioned above, we do require that all pre-models satisfy
the following property:
\begin{description}
\item[(pr)] For each $x\in P\cap R$, we have $x\in V(w)$ if and only
  if $w\in x(w)$.
\end{description}
This says that for propositional letters that are also reasons, the
truth assignment given to $x$ by the valuation $V$ agrees with the
reflexivity of $x$.  This ensures that there is no ambiguity in the
assignment of truth to propositional letters that are also
reasons.\footnote{By definition, the nonempty sets $P$ and $R$ are not
  necessarily disjoint.  Therefore, our language allows for the
  possibility that there are objects that are both propositional
  letters and reasons.  The choice as to which possibility to realize
  is up to the user, who may decide to take $P\cap R=\emptyset$ or not
  as per her preference.}

A \emph{pointed pre-model} is a pair $(M,w)$ consisting of a pre-model
$M$ and a world $w$ in $M$. We write $M,w\models\varphi$ to say that
$\varphi$ is true at the pointed pre-model $(M,w)$, and we write
$M,w\not\models\varphi$ for the negation. We define the satisfaction
relation $\models$ and the set
\[
\sem{\chi}_M:=\{v\in W\mid M,v\models\chi\}
\]
of worlds in the pre-model $M$ at which the formula $\chi$ is
satisfied as follows.
\begin{itemize}
\item $M,w\models p$ means that $p\in V(w)$.

\item $M,w\models\lnot\varphi$ means that $M,w\not\models\varphi$.

\item $M,w\models\varphi\lor\psi$ means that $M,w\models\varphi$ or
  $M,w\models\psi$.

\item $M,w\models r\col\varphi$ means that
  $r(w)\subseteq\sem{\varphi}_M$.

\item $M,w\models r$ means that $w\in r(w)$.

\item $M,w\models B\varphi$ means that $\sem{\varphi}_M\in N(w)$.
\end{itemize}
We note that $\models$ is well-defined: for each $x\in P\cap R$, we
have $x\in V(w)$ if and only $w\in x(w)$ by (pr), and therefore
$M,w\models x$ is well-defined.  

We say that a pre-model $M=(W,[\cdot],N,v)$ is a \emph{model} if and
only if $M$ satisfies each of the following additional properties:
\begin{description}
\item[(brk)] If $\sem{r\col(\varphi\to\psi)}_M\in N(w)$ and
  $\sem{r\col\varphi}_M\in N(w)$, then $\sem{r\col\psi}_M\in N(w)$.

  This says that if the agent believes $r$ supports the implication
  $\varphi\to\psi$ and its antecedent $\varphi$, then she believes
  $r$ supports the consequent $\psi$ as well.

\item[(ba)] If $\sem{r\col\varphi}_M\in N(w)$ and $r^\circ\in N(w)$,
  then $\sem{\varphi}_M\in N(w)$.

  This says that if the agent believes $r$ supports $\varphi$ and she
  believes $r$ is reflexive, then she also believes $\varphi$.

\item[(as)] If $\sem{r\col\varphi}_M\in N(w)$ and $w\in r(w)$, then
  $r(w)\subseteq \sem{\varphi}_M$.

  This says that if the agent believes $r$ supports $\varphi$ and $r$
  is reflexive, then $r$ does in fact support $\varphi$.

\item[(d)] $X\in N(w)$ implies $W-X\notin N(w)$.

  This says that if the agent believes $X$ at world $w$, then she does
  not believe the complement $W-X$ at world $w$.
\end{description}
(brk), (ba), (as), and (d) are the model-theoretic analogs of the
axioms (BRK), (BA), (AS), and (D), respectively. A \emph{pointed
  model} is a pair $(M,w)$ consisting of a model $M$ and a world $w$
in $M$.

Given a pre-model $M$, to say that $\varphi$ is \emph{valid in $M$},
written $M\models\varphi$, means that $\sem{\varphi}_M=W$.  To say
that $\varphi$ is \emph{valid}, written $\models\varphi$, means that
$M\models\varphi$ for every model $M$. That is, validity is taken over
the class of models (and not the larger class of pre-models).  It is
shown in Theorem~\ref{theorem:RBB-determinacy} that $\RBB$ is sound
and complete for this semantics: we have $\vdash\varphi$ if and only
if $\models\varphi$. Unless explicitly noted otherwise, our focus in
what follows will generally be on the concept of model (and not the
concept of pre-model).

The following terminology will be useful for what follows: given a
reason $r$ and a pointed model $(M,w)$ representing the key features
of a particular situation of reason-based belief, to say
\begin{itemize}
\item ``$r$ is adequate at $(M,w)$'' means $M,w\models r$;

\item ``$r$ is veridical for $\varphi$ at $(M,w)$'' means $M,w\models
  r\col\varphi\to\varphi$; and

\item ``$r$ is veridical at $(M,w)$'' means $M,w\models
  r\col\varphi\to\varphi$ for each formula $\varphi$.
\end{itemize}
In using the above terminology, we may omit mention of $(M,w)$ if it
should be clear from context which pointed model is meant.  It follows
by the semantics that every adequate reason is
veridical;\footnote{Proof: suppose $r$ is adequate at $(M,w)$. This
  means $M,w\models r$, which, by the definition of satisfaction,
  means that $w\in r(w)$. Now take an arbitrary $\varphi$. To show
  that $M,w\models r\col\varphi\to\varphi$, assume
  $M,w\models r\col\varphi$. By the definition of satisfaction, this
  assumption means $r(w)\subseteq\sem{\varphi}_M$. Since $w\in r(w)$,
  it follows that $w\in\sem{\varphi}_M$; that is, $M,w\models\varphi$.
  So $r$ is veridical at $(M,w)$.}  however, a reason may be veridical
for a proposition without being thereby adequate. Note that being
veridical is only a necessary condition for a reason to be
adequate. This means we do not take veridicality to provide an
explicit definition of adequacy, but only to put a constraint on what
it is for a reason to be adequate.

\subsection{Weak Belief But Strong Reasons}

$\RBB$ posits a weak notion of belief.  In particular, as reflected in the semantics,
beliefs are not necessarily closed under material implication.  That
is, it is consistent to have
\[
B(\varphi\to\psi)\land B\varphi \land \lnot B\psi\enspace,
\]
which says that the agent believes an implication and the antecedent
of the implication but not the consequent.

Reasons, on the other hand, are strong.  First, as just seen, they
encompass all derivable statements by (RN).  Second, they are closed
under implication (and hence under (MP)) by (RK).  Third, if adequate,
then (A) says that they are \emph{veridical\/}: everything they
support is true.  Reasons therefore support many assertions
(infinitely many, in fact, because each reason supports each of the
infinitely many theorems of $\RBB$ by (RN)).  This puts more
requirements on reasons (i.e., they must do more things), which makes
them stronger.

Reasons are governed by what is essentially the normal multi-modal
logic $\mathsf{KT}$ (with one modal operator ``$\,r\col$'' for each
reason $r$), except that the $\mathsf{T}$ axiom
$r\col\varphi\to\varphi$ (sometimes called ``veridicality'') is not
guaranteed to hold unless, according to (A), we make the additional
assumption that $r$ is adequate.  This way of having a multi-modal
logic like $\mathsf{KT}$ but with the ``modal operator'' $r$ itself a
formula (whose truth implies veridicality) is, to our knowledge,
new.\footnote{The usual way of writing our formula $r\col\varphi$
  would be $\Box_r\varphi$. So $\RBB$ may be viewed as a multi-modal
  logic with an extra formula $r$ for each $r\in R$, a
  $\mathsf{K}$-modality ``$\,r\col$'' (our variant of $\Box_r$) for
  each $r\in R$ that respects the reflexivity scheme $\mathsf{T}$ if
  $r$ holds (as per (A)) and interacts according to (AS) with an
  $\mathsf{ED}$-modality $B$ that is governed by (BRK) and (BA).  Note
  that (AS) provides some interesting interaction between the various
  modal operators. According to our intended semantics
  (\S\ref{section:RBB-semantics}), we interpret modal operators using
  a possible worlds semantics (with a neighborhood function for $B$),
  and the intended interpretation of the formula $r$ is that the
  binary accessibility relation corresponding to the modal operator
  ``$\,r\col$'' is reflexive.  Certain hybrid logics \cite{Bra11:SEP}
  can express reflexivity of modal operators: the formula
  ${\downarrow}x.\lnot\Box_r\lnot x$ says that the accessibility
  relation corresponding to $\Box_r$ is reflexive.  However, hybrid
  logics generally include additional features permitting greater
  expressivity than we need.}

We have chosen a theory of weak belief but strong reasons in order to
keep things as simple as possible but still address some of the major
trends in the epistemological study of knowledge as justified true
belief (``JTB'').  In all of the examples from epistemology we
consider in this paper, we need some way to track an agent's logical
inferences and some way to link these inferences to what the agent
believes.  Our theory $\RBB$ is a rather minimalistic way of doing
just this: reasons are used to handle the relevant inferencing, the
agent can ``accept'' a reason (or not) by believing it to be adequate
(or not), the agent comes to believe things supported by reasons she
accepts, and the agent's beliefs are always consistent.  This way we
can encode inferencing using a reason, indicate whether the agent
accepts this inferencing, and thereby infer whether she believes some
statement based on a reason.  We also allow the possibility that she
believes something without a reason, by which we mean that the set
\[
\{B\varphi\}\cup\{B(r\col\varphi)\to\lnot Br\mid r\in R\}
\]
is consistent with our theory. This assumption is metaphysically
disputable, however, and we may provide for every belief to be
accompanied by a reason if we wish.

\subsection{Consistency of Reasons}

One interesting theorem of $\RBB$ is the principle
\begin{equation*}\tag{RC}
  \vdash(Br\land Bs)\to( B(r\col\varphi)\to\lnot B(s\col\lnot\varphi))
  \label{eq:RC}
\end{equation*}
of \emph{reason consistency}. This principle says that if an agent
believes reasons $r$ and $s$ are adequate, then she cannot believe
that $r$ and $s$ support contradictory assertions. Intuitively, the
derivability of \eqref{eq:RC} follows because (BA) requires that an
agent who accepts a reason and believes that reason to support a
formula must also believe that formula, and (D) requires that an agent
not have contradictory beliefs.  Notice that if we take $r=s$ in
\eqref{eq:RC}, then we obtain a statement provably equivalent to
\begin{equation*}\tag{IC}
  \vdash Br\to(B(r\col\varphi)\to\lnot B(r\col\lnot\varphi))\enspace,
  \label{eq:IC}
\end{equation*}
which says that a reason believed to be adequate is believed to be
internally consistent.

If the agent does not believe $r$ is a adequate reason, then she can
believe that $r$ is internally inconsistent.  Put another way, the
formula
\[
\lnot Br \land B(r\col\varphi) \land B(r\col\lnot\varphi)
\]
is consistent with our theory.  Similarly, if the agent believes $r$
is an adequate reason but does not believe $s$ is an adequate reason,
then our theory does not rule out the possibility that she believes
$r$ and $s$ are inconsistent.  That is,
\[
Br\land \lnot Bs \land B(r\col\varphi)\land B(s\col\lnot\varphi)
\]
is also consistent with our theory.

Finally, since the theory $\RBB$ is consistent (and hence does not
prove both $\varphi$ and $\lnot\varphi$ for some formula
$\varphi$),\footnote{Consistency of $\RBB$ follows by soundness
  (Theorem~\ref{theorem:RBB-determinacy}).}  any adequate reason is
internally consistent.  That is,
\begin{equation*}\tag{AIC}
  \vdash r\col\varphi\to(r\to\lnot (r\col\lnot\varphi))\enspace,
  \label{eq:AIC}
\end{equation*}
which says that a reason $r$ that supports $\varphi$ and is adequate
cannot also support $\lnot\varphi$.  It follows from \eqref{eq:AIC}
that any internally inconsistent reason is not adequate.

\subsection{Logical Closure and Combination of Reasons}
\label{section:closure}

If $\psi$ is a logical consequence of $\varphi$, meaning
$\vdash\varphi\to\psi$, then our theory says that $r$ is a reason to
believe the consequent $\psi$ whenever $r$ supports the
antecedent $\varphi$.  That is,
\begin{equation*}\tag{RCLC}
  \vdash \varphi\to\psi
  \quad\text{implies}\quad
  \vdash r\col\varphi\to r\col\psi\enspace.
  \label{eq:RCLC}
\end{equation*}
The proof: from $\varphi\to\psi$, we obtain $r\col(\varphi\to\psi)$ by
(RN).  This is the antecedent of an instance of (RK), so the
consequent $r\col\varphi\to r\col\psi$ of this instance is derivable
by an application of (MP). This completes the proof.  In examining
this proof, we see that \eqref{eq:RCLC} is a consequence of the
stronger logical principle encompassed by our axiom
\begin{equation*}\tag{RK}
  r\col(\varphi\to\psi)\to(r\col\varphi\to r\col\psi)\enspace,
\end{equation*}
which says that reasons are closed under material implication.

The principle \eqref{eq:RCLC} says that reasons are closed under
logical consequence.  It is unclear whether this is a desirable
principle. For example, it may make more sense to say that if $\psi$
is a logical consequence of $\varphi$ and $r$ supports
$\varphi$, then it is not $r$ itself that supports the
consequence $\psi$.  Instead, what is wanted is some more complicated
reason $r'$ that not only references $r$ but also provides some reason
$s$ as to why $\psi$ obtains from $\varphi$.  That is, we might seek a
principle like this:
\begin{equation*}\tag{App}
  s\col(\varphi\to\psi)\to(r\col\varphi\to(s\cdot r)\col\psi)
  \enspace.
  \label{eq:App}
\end{equation*}
This is the principle of \emph{Application} from Justification Logic
\cite{ArtFit11:SEP}.  It says: if $s$ supports the
implication $\varphi\to\psi$ and $r$ supports the
antecedent $\varphi$, then a new object $s\cdot r$ obtained by
combining $s$ and $r$ supports the consequent $\psi$.
In essence, the more complex reason $s\cdot r$ keeps track of
everything we would need to check to see that $\psi$ indeed obtains:
the initial reason $r$ for the antecedent $\varphi$ and a reason $s$
for the implication $\varphi\to\psi$.  Further, the syntactic
structure of $s\cdot r$, with $s$ to the left and $r$ to the right,
tells us what kind of a reason we have: based on the form of
\eqref{eq:App}, it is suggested that $s$ is some implication, $r$ is
the antecedent of that implication, and $s\cdot r$ is a reason for the
consequent.  In essence, we are describing specific witnesses for an
instance of the rule (MP) of Modus Ponens:
\[
\frac{\varphi\to\psi\quad\varphi}{\psi}\;\text{\footnotesize(MP)}
\quad\text{is encoded by}\quad
\frac{s\quad r}{s\cdot r}\enspace.
\]

\eqref{eq:App} is a more nuanced version of (RK): if we have
$r\col(\varphi\to\psi)$ and $r\col\varphi$, then we do not obtain
$r\col\psi$ straightaway using \eqref{eq:App}.  Instead, we must
construct the reason $r\cdot r$ in support of $\psi$.  The single
instance of ``$\,\cdot\,$'' in the syntactic structure of the latter
reason reflects our use of one derivational step (i.e., one instance
of (MP)) to obtain $\psi$.

To do away with \eqref{eq:RCLC}, we must do away with (RK) and modify
(RN).  In particular, let $R_0$ be a nonempty set of ``basic''
reasons, define $R$ to be the smallest extension of $R_0$ satisfying
the property that $s,t\in R$ implies $s\cdot t\in R$, replace scheme
(RK) by scheme \eqref{eq:App}, and restrict (RN) by requiring that
$r\in R_0$. (It is assumed that all other schemes and rules can use
reasons coming from the full set $R$.) Call the resulting theory
$\RBB{+}\text{\eqref{eq:App}}$.  In $\RBB{+}\text{\eqref{eq:App}}$, it
is consistent to have
\begin{equation}\label{eq:notRCLC}
  r\col(\varphi\to\psi)\land r\col\varphi\land\lnot (r\col\psi)
  \enspace,
\end{equation}
which says that $r$ is not closed under implication. As a result, it
can be shown that \eqref{eq:RCLC} no longer obtains.  But note that
(RK) and \eqref{eq:RCLC} do not fail in $\RBB{+}\text{\eqref{eq:App}}$
because logical or materially implied consequences of assertions are
no longer ``accessible'' by some reason.  Indeed, in the situation
\eqref{eq:notRCLC} under the theory $\RBB{+}\text{\eqref{eq:App}}$,
the logical consequence $\psi$ of $\varphi$ is still ``accessible'' by
the reason $r\cdot r$.  However, this reason $r\cdot r$ is more
``complex'' than the original reason $r$ (in terms of the number of
instances of the (MP)-encoding Application operator ``\,$\cdot$\,''
that appear inside it).  In general, distant consequences that would
require many repetitions of \eqref{eq:App} are still accessible; it is
just that the reasons that access these consequences may be very
``complex.''

Justification Logic (JL) \cite{ArtFit11:SEP} is the study of logics of
reason-based belief (with reasons thought of as ``justifications'').
Defining $\JL_0$ to be the fragment of $\RBB{+}\text{\eqref{eq:App}}$
obtained by omitting all belief formulas and belief axioms from the
theory, Justification Logics may be thought of as extensions of
$\JL_0$.\footnote{Actually, Justification Logics are extensions of a
  fragment of $\RBB{+}\text{\eqref{eq:App}}$ that places further
  restrictions on the rule (RN), but we set aside further discussion
  of this issue here.}  Many JLs permit other kinds of combinations of
reasons than what we saw with \eqref{eq:App}.  Our logic $\RBB$ is
closely related to the JL tradition, though we conspicuously omit
\eqref{eq:App}, retain (RK) (and thereby endorse \eqref{eq:RCLC}),
leave (RN) without the additional restriction, and maintain a set $R$
of primitive reasons that cannot be combined to form more complex
reasons.  In so doing, we lose the ability to have the syntactic
structure of reasons reflect the structure of derivations in the
theory, and thereby forgo a more nuanced tracking of the interaction
between logical consequence and the complexity of reasons.  We accept
these consequences in the interest of developing a system that is
simple and yet still of use to the formal epistemologist.
Nevertheless, we recognize that a reader may be interested to see a
thorough study of more sophisticated extensions of our theory that
allow for the combination of reasons along the lines of \eqref{eq:App}
(and perhaps for other features as well).  We advise such a reader to
consult the JL literature directly \cite{ArtFit11:SEP}.

Since our theory does not allow the agent to combine reasons in the
sense of \eqref{eq:App} and beliefs are not closed under implication,
it is consistent for us to have a situation wherein the agent has the
requisite information to draw a belief but simply does not draw
it. For example, assuming $s\neq r$, it is consistent to have
\begin{equation}
  Bs\land Br \land B(s\col(\varphi\to\psi))\land B(r\col\varphi) \land \lnot
  B\psi\enspace,
  \label{eq:noRCL}
\end{equation}
which says that the agent believes $s$ is adequate and believes $s$
supports an implication (and hence believes the implication by (BA)),
believes $r$ is adequate and believes $r$ supports the antecedent of
the implication (and hence believes the antecedent by (BA)), and yet
the agent does not believe the consequent (even though the believes
the implication and its antecedent).  The problem is that her beliefs
are not closed under (MP). This is so despite the fact that, by
\eqref{eq:RCLC} and (RB), beliefs coming from \emph{the same reason}
are closed under (MP):
\[
\vdash (Bt\land B(t\col(\varphi\to\psi))\land B(t\col\varphi))\to
B\psi\enspace.
\]
So long as the implication and its antecedent come from
\emph{separate} beliefs as in \eqref{eq:noRCL}, the agent need not
believe the consequent. 

The consistency of \eqref{eq:noRCL} is a consequence of our design: we
use reason operators to encode inferencing, and we use belief
operators to encode the particular inferencing and the individual
assertions that the agent accepts (in terms of her affirmed beliefs).
As such, the situation \eqref{eq:noRCL} is one in which the agent has
not yet performed sufficient inferencing to accept the conclusion
$\psi$, even though she is very close (after all, she has done enough
to accept the implication $\varphi\to\psi$ and its antecedent
$\varphi$).  In essence, this lack of closure allows us to place one
kind of constraint on the agent's inferencing powers.  If desired, one
can place even more severe constraints as in \cite{BalRenSme14:APAL};
however, this seems to require more syntax and additional axioms.  One
can also go the other way and lift these constraints entirely; in
\S\ref{section:sigma} we will suggest one natural way to do this in
our setting.  But for now we retain what we hope is a ``happy medium''
in the form our theory $\RBB$.

\subsection{Implication-Closed and Purely Reason-Based Beliefs}
\label{section:sigma}

We saw in \eqref{eq:noRCL} that reason-based beliefs in $\RBB$ need
not be closed under implication if the source reasons are different.
If we would like to ensure that reason-based beliefs are always closed
under implication, even if the beliefs come from separate reasons,
then a simple solution is to introduce a ``master reason'' $\sigma$
that encodes the combined information of all reasons the agent
accepts.  This requires us to expand our reason set $R$ to include a
new symbol $\sigma$ not already present and then add the following
additional schemes to $\RBB$:
\begin{eqnarray*}
  & \text{\bf(MA)} & \sigma\to(Br\to r) \\
  & \text{\bf(MB)} & B\sigma \\
  & \text{\bf(MR)} & B(r\col\varphi)\to(Br\to B(\sigma\col\varphi))
\end{eqnarray*}
(MA) says that every accepted reason is adequate if the master reason
is adequate, (MB) says that the agent always accepts the master
reason, and (MR) says that anything the agent believes is supported by
an accepted reason the agent will also believe to be supported by the
master reason.  Adding these principles makes it so that $\sigma$ is
the sum of the agent's evidence.  Calling $\RBB_\sigma$ the theory
obtained from $\RBB$ by expanding the set $R$ of reasons to include a
new symbol $\sigma$ and adding (MA), (MB), and (MR), it follows that
\begin{equation*}\tag{RCL2}
  \RBB_\sigma
  \vdash (Bs\land Br\land B(s\col(\varphi\to\psi)) \land B(r\col\varphi))
  \to
  B\psi\enspace.
  \label{eq:RCL2}
\end{equation*}
Indeed, if the agent accepts $s$, believes $s$ supports an
implication, accepts $r$, and believes $r$ supports the antecedent of
the implication, then it follows by (MR) that the agent believes
$\sigma$ supports the implication and its antecedent.  Applying (BRK),
it follows that the agent believes $\sigma$ supports the
consequent. But the agent also accepts $\sigma$ by (MB), so it follows
by an application of (BA) that the agent believes the consequent.  It
is in this sense that the agent ``combines'' the information conveyed
by reasons the agent accepts into the master reason $\sigma$.

While we have shown that the $\RBB_\sigma$-agent may combine the
information from two reasons to derive her beliefs, there is no need
to restrict the number of reasons to two.  Indeed, according to (MB),
the agent implicitly combines into $\sigma$ the information from
\emph{every} reason she believes to be adequate, no matter how many of
these there may be.

In $\RBB_\sigma$, the master reason $\sigma$ serves as a witness for
an existential quantifier over the believed support of accepted
reasons.\footnote{Said precisely: $\sigma$ is a Skolem constant for
  the existential quantifier over the believed support of accepted
  reasons. See \cite{vDal94:Book} or any book on mathematical logic
  for details.}  In particular, (MR) tells us that if there exists an
accepted reason $r$ that the agent believes supports $\varphi$, then
the agent believes that $\sigma$ supports $\varphi$.  Hence by (MB),
if there exists such an $r$, then the agent believes $\sigma$ supports
$\varphi$ and is itself accepted.  If we were to add quantifiers to
the language (something we do later in \S\ref{section:QRBB}), we could
express this as:
\[
(\exists r)(Br\land B(r\col\varphi))\to(B\sigma\land B(\sigma\col\varphi))
\enspace.
\]
That is, if there exists an accepted reason that the agent believes
supports $\varphi$, then $\sigma$ is a witness to the existential
quantifier.

Both $\RBB_\sigma$ and the basic theory $\RBB$ allow for the
possibility that the agent believes a formula $\varphi$ without any
supporting reasons (i.e., she does not believe to be adequate any
reason that she believes supports $\varphi$). This is the same as
saying that the set
\[
\{B\varphi\} \cup \{ B(r\col\varphi)\to\lnot Br \mid r\in R\}
\]
is consistent with both $\RBB_\sigma$ and $\RBB$.  If this situation
is undesirable, then a simple remedy is to extend the theory
$\RBB_\sigma$ by adding a principle that says all beliefs are believed
to be supported by the master reason:
\begin{eqnarray*}
  & \text{\bf(MT)} & B\varphi\to B(\sigma\col\varphi) 
\end{eqnarray*}
With (MT) in place, the agent believes only those things she believes
are supported by $\sigma$ and hence, by (MR), she believes only those
things she believes are supported by some reason.  In short, every
belief is ``reason-based.''  Defining $\RBB_\sigma^+$ to be the theory
obtained from $\RBB_\sigma$ by adding (MT), it follows by (MT), (RB),
and (MB) that
\[
\RBB_\sigma^+\vdash B\varphi \leftrightarrow
B(\sigma\col\varphi)\enspace,
\]
which says that the agent believes something just in case it is
supported by the master reason. But then belief can be conflated with
the master reason.  As a result, we have by \eqref{eq:RCLC} that the
beliefs of the $\RBB_\sigma^+$-agent are always closed under (MP).
Belief in $\RBB_\sigma^+$ is therefore governed by the normal modal
logic $\mathsf{KD}$.

Semantics for $\RBB_\sigma$ and for $\RBB_\sigma^+$ may be found in
\S\ref{section:sigma-semantics}. It is shown in
Theorem~\ref{theorem:RBBsigma-determinacy} that each of $\RBB_\sigma$
and $\RBB_\sigma^+$ is sound and complete for its semantics.

\subsection{Quantification Over Reasons}
\label{section:QRBB}

We have observed that the theory $\RBB$ does not require that every
belief arise from a reason: it is consistent with $\RBB$ for the agent
to believe $\varphi$ and yet have no accepted reason $r$ she believes
supports $\varphi$.  If we were to introduce quantifiers over reasons
into our language, then we could express this situation by saying that
the following formula is consistent:
\[
B\varphi \land (\forall r)(B(r\col\varphi)\to\lnot Br)\enspace.
\]
Another example: we might like to say that $r$ is the unique accepted
reason the agent believes supports $\varphi$:
\[
B(r\col\varphi) \land Br \land (\forall s)( (s\neq r \land B(s\col\varphi))
\to \lnot Bs)\enspace.
\]
To allow such expressions as formulas, we extend our set of formulas
$F$ to the larger set $F^\forall$ consisting of all formulas $\varphi$
that may be formed by the following grammar:
\begin{eqnarray*}
  \varphi &::=&
  p \mid \lnot\varphi \mid (\varphi\lor\varphi) \mid
  (r\col\varphi) \mid r \mid B\varphi
  \mid r=r \mid (\forall r)\varphi
  \\
  &&
  \text{\footnotesize
    $p\in P$, $r\in R$}
\end{eqnarray*}
We adopt usual Boolean connective abbreviations along with the
following:
\begin{eqnarray*}
  r\neq s &:=& 
  \lnot(r=s) 
  \\
  (\exists r)\varphi &:=& 
  \lnot(\forall r)\lnot\varphi 
  \\
  (\forall r\neq s)\varphi &:=&
  (\forall r)(r\neq s\to\varphi)
  \\
  (\exists r\neq s)\varphi &:=&
  (\exists r)(r\neq s\land\varphi)
\end{eqnarray*}

Note that in this language, an element $r\in R$ can act both as a
reason (as in the formula $r\col p$) and as a quantifier variable (as
in the formula $(\forall r)(r\col p)$).  Therefore, reasons may appear
either bound or free in formulas, with the notion of \emph{bound} and
\emph{free} defined in the usual way.  For reasons $s$ and $r$ and a
formula $\varphi$, we say that \emph{$s$ is free for $r$ in $\varphi$}
to mean that $r$ has no free occurrence in $\varphi$ within the scope
of a quantifier $(\forall s)$.  Put another way, if $s$ is free for
$r$ in $\varphi$, then in the formula
\[
\varphi[s/r]
\]
obtained by substituting all free occurrences of $r$ in $\varphi$ by
$s$, no newly replaced occurrence is bound.  Examples: $s$ is free for
$r$ in $(\forall t)(t\neq r)$ but not in $(\forall s)(s\neq r)$.

The theory $\QRBB$ of \emph{Quantified Reason-Based Belief} is defined
by adding to the axiomatization of $\RBB$ the axioms and rules in
Table~\ref{table:QRBB}.  (UD), (UI), and (Gen) are standard principles
of first-order quantification.  (EP) and (EN) say that two reasons are
considered to be the same if and only if they are syntactically
identical.

\begin{table}[ht]
  \begin{center}
    \textsc{Additional Axiom Schemes}\\[.4em]
    \renewcommand{\arraystretch}{1.3}
    \begin{tabular}[t]{rl}
      (UD) &
      $(\forall r)(\varphi\to\psi)\to
      (\varphi\to(\forall r)\psi)$, where $r$ is not free in $\varphi$
      \\
      (UI) &
      $(\forall r)\varphi\to\varphi[s/r]$, where
      $s$ is free for $r$ in $\varphi$
      \\
      (EP) &
      $r=r$
      \\
      (EN) &
      $\lnot(r=s)$, where $r$ and $s$ are syntactically different
    \end{tabular}
    \renewcommand{\arraystretch}{1.0}
    \\[1em]
    \textsc{Additional Rule}\vspace{-.5em}
    \[
    \frac{\varphi}{(\forall r)\varphi}
    \;\text{\footnotesize(Gen)}
    \]
  \end{center}
  \caption{The theory $\QRBB$ consists of $\RBB$ augmented with the above axioms and rule}
  \label{table:QRBB}
\end{table}

It is shown in Corollary~\ref{corollary:extension} that for each
$\varphi\in F$ not containing quantifiers, we have
$\QRBB\vdash\varphi$ if and only if $\RBB\vdash\varphi$.  It is
therefore unproblematic for us to simply write $\vdash\varphi$ to say
that $\varphi$ is provable.

We can extend $\QRBB$ to the theory $\QRBB_\sigma$ obtained by
extending $R$ to include a new master reason $\sigma$ and adding the
schemes (MA), (MB), and (MR) for $\sigma$.  We can further extend
$\QRBB_\sigma$ to the theory $\QRBB_\sigma^+$ obtained by adding the
additional scheme (MT) to guarantee all beliefs are reason-based.

Semantics for $\QRBB$, for $\QRBB_\sigma$, and for $\QRBB_\sigma^+$
may be found in \S\ref{section:QRBB-semantics}. It is shown in
Theorems~\ref{theorem:QRBB-determinacy} and
\ref{theorem:QRBBsigma-determinacy} that each of $\QRBB$,
$\QRBB_\sigma$, and $\QRBB_\sigma^+$ is sound and complete for its
semantics. However, for the completeness results, there is one caveat:
our proofs require that the set $R$ of reasons be at least countably
infinite.

\section{Justified True Belief and Knowledge}
\label{section:jtb}

We use our logical framework to tease apart two notions of justified
true belief (henceforth ``JTB'').  The first is an internalist notion,
which Gettier showed was insufficient for knowledge \cite{Get63:A}.
The second is an externalist notion that we argue is immune to Gettier
scenarios. More generally, we show that our framework can distinguish
three ``types'' of reasons: those that are non-veridical (and hence
inadequate), those that are veridical for the proposition they
support, but inadequate, and those that are adequate (and hence
veridical).  Gettier's second case has reasons of the first type:
non-veridical.  The ``fake barn county'' case has reasons of the
second type: veridical for the proposition they support, but
inadequate. Other cases (such as ``normal barn county'', or ``good
cases'' more broadly, see \cite{Wil00:Book}) have reasons of the third
type: adequate.

\subsection{Two Notions of Justification}

In our theory, there are (at least) two natural ways to define JTB:
\begin{itemize}
\item $\JTB^e_r(\varphi):=B(r\col\varphi)\land Br\land r$, and

\item $\JTB^i_r(\varphi):=B(r\col\varphi)\land Br\land\varphi$.
\end{itemize}
Both imply that the agent has a true belief that $\varphi$.  However,
$\JTB^e_r(\varphi)$ suggests that the agent has a true belief
justified by an adequate reason, whereas $\JTB^i_r(\varphi)$ suggests
that the agent only has a true belief justified by a \emph{prima
  facie} reason (that may not be adequate).  $\JTB^e_r$ is thus
externalist, while $\JTB^i_r$ is internalist.

Gettier's achievement was to deny that $\JTB^i_r(\varphi)$ is the same
as knowledge of $\varphi$.  Thus, if we assume that
\begin{equation*}\tag{G2}
  B(r\col p) \land Br \land B(r\col (p \rightarrow p \vee q)) \land (\lnot p\land q)\enspace,
  \label{eq:G2}
\end{equation*}
then we have Gettier's second case. This is the case where the agent
named Smith has a reason to believe $p$ (``Jones owns a Ford'') and
``realizes'' that $p\lor q$ (``Jones owns a Ford or Brown is in
Barcelona'') follows from $p$ on that basis; however, unknown to
Smith, $p$ is false and $q$ is true, hence $p \vee q$ is premised on a
false lemma. Let us call $r$ the reason Smith has to believe $p$. We
can safely assume that $r$ does indeed support $p$ in that case
($r\col p$) (i.e., the past evidence adduced by the agent does
correspond to a real experience of his) and we leave this premise
implicit in \eqref{eq:G2} above, since delusion is not the problem in
this case.  Since $p$ is false, and $r\col p$ by assumption, $r$
cannot be adequate, by (A). However, since $r$ supports $p$, we have
by \eqref{eq:RCLC} that $r$ supports $p\lor q$.  Smith is said by
Gettier to realize that $p$ entails $p\vee q$ on the basis of his
reason, and by (BRK) it follows that $B(r\col(p\vee q))$. Moreover,
Smith has no reason supporting $q$ that she believes is adequate
(indeed, in the scenario, ``Brown is in Barcelona'' is consciously
picked at random by Smith). So we are in a situation where Smith has
an internally justified true belief that $p$ ($\JTB^i_r p$), and also
an internally justified true belief that $p \vee q$
($\JTB^i_r(p\lor q)$), but he fails to have an externally justified
true belief of either proposition ($\neg \JTB^e_r p$, and
$\neg \JTB^e_r(p\lor q))$.

In contrast, if we assume that $\JTB^e_r(p)$, which is
\begin{equation*}\tag{G2$'$}
  B(r\col p)\land Br \land r\enspace,
  \label{eq:G2prime}
\end{equation*}
this time $r$ is adequate. Since $r$ supports $p$, it also supports
$p\lor q$ by \eqref{eq:RCLC}.  By (A) both $p$ and $p\lor q$ are true.
Since the agent believes $r$ is an adequate reason supporting $p$ (and
therefore also supporting $p\lor q$), she believes both $p$ and
$p\lor q$, and in this case her belief is based on an adequate (and
hence veridical) reason.

In general, it is easy to see that $\JTB^e_r(\varphi)$ satisfies:
\begin{itemize}
\item $\vdash \JTB^e_r(\varphi\to\psi)\to(\JTB^e_r(\varphi) \to
  \JTB^e_r(\psi))$,

  which says that external JTB based on a reason $r$ is closed under
  implication;

\item $\vdash\JTB^e_r(\varphi)\to \varphi$,

  which says that external JTB is veridical; and

\item $\vdash\JTB^e_r(\varphi)\to(r\col\psi\to\psi)$,
  
  which says that if an agent has an external JTB based on reason $r$,
  then $r$ cannot support any false assertions (so-called ``false
  lemmas'').
\end{itemize}
To compare with internal JTB, one can show that $\JTB^i_r(\varphi)$ satisfies:
\begin{itemize}
\item $\vdash \JTB^i_r(\varphi\to\psi)\to(\JTB^i_r(\varphi) \to
  \JTB^i_r(\psi))$,

  which says that internal JTB based on a reason $r$ is closed under
  implication;

\item $\vdash\JTB^i_r(\varphi)\to \varphi$,

  which says that internal JTB is veridical; and

\item $\nvdash\JTB^i_r(\varphi)\to(r\col\psi\to\psi)$,
  
  which says that if an agent has an internal JTB based on reason $r$,
  then $r$ might support false assertions (so-called ``false
  lemmas'').
\end{itemize}
The differences between $\JTB^e_r$ and $\JTB^i_r$ are in the last
property. So we see that the main difference between external and
internal JTB is in the reliability of the reason on which the JTB is
based.  

Using our quantified language, we adopt the following abbreviations:
\[
\JTB^e(\varphi):=(\exists r)\JTB^e_r(\varphi)
\quad\text{and}\quad
\JTB^i(\varphi):=(\exists r)\JTB^i_r(\varphi)
\enspace.
\]
$\JTB^e(\varphi)$ says that the agent has an external JTB for
$\varphi$ (based on some reason), and $\JTB^i(\varphi)$ says the same
but for internal JTB. 

\subsection{Is Knowledge $\JTB^e$?}

$\JTB^i$ falls prey to Gettier's examples because the supporting
reason need not be veridical (i.e., it admits ``false
lemmas''). $\JTB^e$, however, requires an adequate supporting reason,
and hence this reason is necessarily veridical (i.e., it admits ``no
false lemmas'').  This suggests we examine the equation
\begin{equation*}\tag{KJTBe}
  K\varphi := \JTB^e(\varphi)\enspace,
  \label{eq:knowledge}
\end{equation*}
which defines knowledge as external JTB.  What should we think of this
equation? 


Consider the ``fake barn county'' situation \cite{Gol76:JP}:
the agent is in a county that has numerous fake barns that look
exactly like real barns.  Not knowing she is in this county, she sees
what she thinks is a barn and concludes that it is indeed a barn.  It
turns out she is correct because, by chance, she happens to be looking
at the only real barn in the entire county.  Obviously, she has an
internal JTB that she sees a barn, though most philosophers argue that
she does not know she sees a barn.\footnote{Intuitions about knowledge
  ascription in fake barn cases are notoriously less stable among
  philosophers than they are in the original Gettier cases (see
  \cite{Kra02:LP,Tur12:S,egre2017knowledge}). Here we are considering a
  situation in which the belief seems simply ``too lucky'' to count as
  knowledge.} But does she have an external JTB in this case? One
reason to answer affirmatively: the agent's reason is veridical,
unlike in Gettier's original examples.

However, veridicality does not imply adequacy. Take $r$ and $p$ so
that
\begin{eqnarray*}
  r & \text{is read,} &
  \text{``the agent sees what looks to her like a barn,'' and}
  \\
  p & \text{is read,} &
  \text{``what the agent sees is a barn.''}
\end{eqnarray*}
Our agent is in the situation:
\begin{equation*}\tag{Barn}
  B(r\col p)\land Br \land p\enspace.
  \label{eq:Barn}
\end{equation*}

\noindent That is, the agent believes that her seeing what looks like
a barn supports the assertion that what she sees is a barn, she
believes $r$ is adequate to guarantee the truth of what it supports,
and the agent does actually see a barn. But is $r$ in fact adequate?
If we say it is, then we run into the following problem: had the agent
picked a different barn-looking structure that turned out to be a
fake, we would have
\begin{equation*}\tag{Barn$'$}
  B(r\col p)\land Br\land\lnot p\enspace,
  \label{eq:Barn'}
\end{equation*}
from which it would follow by (AS) and the assumed adequacy of $r$ that $p$ holds 
(i.e., we have $(r\col p)\land p$), but this
contradicts the assumed adequacy of $r$ (since in fact $\neg p$). This suggests to us that $r$ is not necessarily
adequate; that is, each of $\eqref{eq:Barn}\land r$ and
$\eqref{eq:Barn}\land\lnot r$ is consistent with our intuitions about
the ``fake barn county'' example.  Conclusion: the agent need not have
external JTB in this case.

We take it that the ``fake barn county'' example seeks to challenge
the agent's acumen in determining when it is safe to reason according
to the principle
\begin{equation*}\tag{WSWG}
  \text{(what I see looks like an $X$)}
  \to
  \text{(what I see is an $X$)}\enspace,
  \label{eq:WSWG}
\end{equation*}
which has the colloquial reading ``what I see is what I
get.''\footnote{This diagnosis was suggested to the third author by
  Alexandru Baltag (private communication).}  Since adequacy implies
veridicality, one could use our notion of adequacy to indicate that
the agent's use of \eqref{eq:WSWG} is licensed.  In particular, if we
assume that $\eqref{eq:Barn}\land r$, then we might construe this as a
case in which the agent is in ``normal barn county'' (where there are
no fake barns) and so her use of \eqref{eq:WSWG} is licensed: $r$ is
an externally valid reason for the agent to infer that she sees a
barn, so the agent knows that she sees a barn.  In contrast, if we
assume that $\eqref{eq:Barn}\land\lnot r$, then we might construe this
as a case in which the agent is back in ``fake barn county'' and not
licensed to draw the conclusion: $r$ is not an externally valid reason
for her to infer that she sees a barn, so she does not know that she
sees a barn.

So let us assume that our definition \eqref{eq:knowledge} of knowledge
as external JTB is correct. Is it then possible to define knowledge
(i.e., external JTB) in terms of internal JTB plus some other
condition?  Indeed it is:
\begin{equation*}
  \vdash \JTB^e(\varphi)  
  \leftrightarrow (\exists r)(r \land \JTB^i_r(\varphi))
\end{equation*}
In words: to have external justification it is necessary and
sufficient to have an adequate justification that serves as the basis
for an internal JTB. Zagzebski's criticism of a JTB-based analysis of
knowledge \cite{Zag94:PQ} might apply here: we either must sever the
link between truth and justification (thereby going so far as to
concede that there is knowledge in Gettier cases) or else assert that
``there is no degree of independence at all between truth and
justification'' (in order to avoid Gettier problems).  Zagzebski's
position is that neither horn of her dilemma is satisfactory, and so
the proper way to avoid the dilemma is to reject the possibility of
analyzing knowledge in terms of JTB plus some extra component (i.e.,
reject the ``knowledge is $\text{JTB}+x$'' approach all together); see
also \cite{Wil00:Book}.  We accept that our approach is close to
endorsing the second horn of Zagzebski's dilemma. However, by
distinguishing adequacy from veridicality, we can still maintain a
notion of independence between truth and justification. In particular,
\emph{pace} Zagzebski, our semantic analysis distinguishes between
``adequate belief'' (i.e., $\JTB^e$) and ``lucky true belief'' (i.e.,
$\JTB^i$).

\section{Infallibilism and the Problem of Mixed Reasons}\label{section:mixed}


Our analysis commits us to an infallibilist view of knowledge. Dutant
in \cite{Dut10:PHD,Dut15:Method} defines \emph{method infallibilism}
as for an agent to know that $p$ if the agent believes that $p$ on the basis
of a method that could only yield true beliefs. Our notion of
knowledge in terms of external JTB achieves the same result: for an
agent to know $p$ is to believe $p$ on the basis of an adequate reason
$r$, hence to believe $p$ on the basis of a reason that could only
support true propositions.\footnote{See Neta's \cite{Neta11:refutation}, and Dutant's
  \cite{Dut10:PHD,Dut15:Method} for an in-depth discussion of
  infallibilism. We leave a comparison of our approach with Neta's and Dutant's respective approaches
  for another occasion.}  In this section we propose to discuss two
specific issues concerning our definition of knowledge in terms of
$\JTB^e$. Both issues raise the problem of whether a definition of
knowledge as $\JTB^e$ might be either too weak, or too strong,
depending on the case.

\subsection{Knowledge from mixed reasons?}

Let us start with the worry that our account might be too weak. The
$\JTB^e$ analysis of knowledge raises the issue of the force of the
quantifier on the right side of the equivalence. To appreciate the
problem, it is worth reminding ourselves of one of the first responses
to Gettier's examples: the so-called ``no false lemmas'' (hereafter
``NFL'') requirement (see \cite{Cla63:A,Har70:KRC,Sos74:how}). The NFL
requirement is meant to rule out situations like Gettier's second
case, wherein the agent starts from a mistaken belief that $p$ to
obtain a correct belief that $p\lor q$.  Thinking of the reasoning
sequence of beliefs $\may{p,p\lor q}$ as a ``proof,'' the initial
``lemma'' (i.e., assumption) $p$ is false, but then a perfectly
legitimate inference step to a logical consequence $p\lor q$ ends up
on a formula that just so happens to be true.

In our framework, the obvious counterpart to the NFL requirement is
the ``no inadequate lemmas'' (henceforth ``NIL'') requirement:
\[
  \NIL(\varphi) := (\forall s)(\JTB^i_s(\varphi)\to s)\enspace.
  \label{eq:NIL}
\]
This says that every reason that supports an internal JTB of $\varphi$
is adequate.  Since adequate reasons support only true formulas (by
axiom scheme (A)), the NIL requirement guarantees that no false
``lemma'' (i.e., formula) intrudes on a reason justifying a potential
internal JTB of $\varphi$.  This gives rise to the following notion of
JTB with no inadequate lemmas:
\[
\JTB{+}\NIL(\varphi) := \JTB^i(\varphi)\land\NIL(\varphi)\enspace.
\]
This notion of JTB is logically stronger than external JTB:
$\JTB^e(\varphi)$ only requires that there be one adequate reason
supporting an internal JTB of $\varphi$, whereas $\JTB{+}\NIL$
requires adequacy of \emph{every} reason supporting an internal JTB of
$\varphi$.  Thus
\[
\vdash\JTB{+}\NIL(\varphi)\to\JTB^e(\varphi)
\]
but not the other way around.

These considerations raise a potential worry for the $\JTB^e$ analysis
of knowledge: what happens when the agent rests her beliefs in a
proposition (such as $p\lor q$) on multiple sources?  For
example,\footnote{The example was suggested to the first author by
  Timothy Williamson (private communication).  See \cite{egre2017knowledge}.}
suppose our agent, who has excellent eyesight, sees someone in the
distance but cannot quite make out who it is.  Nevertheless, based on
what she can see (represented by reason $s$), she correctly believes
that the person in the distance is either Tweedle Dee or Tweedle Dum
(represented respectively by $p\lor q$).  Further, she has another
reason $r$ to believe that the person is Tweedle Dee (i.e., $p$).  For
example, a friend might have told her that Tweedle Dum is on vacation
in some faraway country. Now, unknown to our agent, the person in the
distance is actually Tweedle Dum. Put formally:

\begin{align*}
  B(s\col(p\lor q))\land \lnot B(s\col p) \land \lnot B(s\col q)
  \land B(r\col p) \;\land 
  \notag{}
  \\
  \tag{TDTD}
  Bs \land Br \land  
  s \land \lnot r\land (\lnot p\land q)
  \enspace.
  \label{eq:TDTD}
\end{align*}
That is, the agent believes $s$ supports the disjunction (that the
person is Dee or Dum) but no disjunct, she believes $r$ supports the
claim it is Dee, she believes $s$ and $r$ to be adequate, $s$ is
adequate (by hypothesis, because the agent's eyesight is excellent,
and it could not possibly be someone other than Dee or Dum), $r$ is
inadequate (since the friend's information is not reliable), and the
person is actually not Dee but Dum.  Now suppose we add to
\eqref{eq:TDTD} the assumption
\begin{equation*}\tag{NoR}
  (\forall t)((t\neq s \land t\neq r)\to \lnot Bt)
  \label{eq:NoR}
\end{equation*}
that the agent believes no other reasons to be adequate. It can be
shown that
\[
\begin{array}{ll}
  \vdash[\eqref{eq:TDTD}\land\eqref{eq:NoR}]\to\JTB^e(p\lor q)
  & \text{but }
  \\[.5em]
  \nvdash[\eqref{eq:TDTD}\land\eqref{eq:NoR}]\to\JTB{+}\NIL(p\lor q)
  \enspace.
\end{array}
\]
In words: the agent has external JTB that $p\lor q$ (because the
adequate reason $s$ supports the disjunction); however, she does not
have JTB with NIL of $p\lor q$ (because the inadequate reason $r$
supports the disjunction). But is it a mistake to equate knowledge
with $\JTB^e$ instead of with $\JTB{+}\NIL$?

One reaction is to deny that there is knowledge when the universal
condition is not satisfied. For an example supporting this reaction,
suppose the agent proves that a certain Mersenne number $m=2^n-1$ is
prime.  Later, she bolsters her belief in the primality of $m$ by
coming to believe (incorrectly) that all Mersenne numbers are prime
(i.e., all numbers of the form $2^k-1$ are prime, which is false).
Can the agent still be said to know that $m$ is prime?  On one
account, it seems not.  Such situations of \emph{mixed reasons}, where
an agent has both adequate and inadequate reasons supporting the same
proposition, arguably occur often in everyday life.

We are inclined to the opposite view: in a situation of
mixed reasons, the agent can still have knowledge.  Returning to the
primality example, if the agent learns that not all Mersenne numbers
are prime, then she will still believe that $m$ is prime on the basis
of her adequate ``backup'' reason (that she proved $m$ is prime).  So
she could still be said to know that $m$ is prime.\footnote{These
  ideas are related to the defeasibility theory of knowledge
  \cite{LehPax69:JP,Leh00:Book}.} The Dee/Dum case is arguably
similar: if the agent were to learn that $r$ is unreliable, then she
would still have an external JTB of the disjunction based on the
``backup'' reason $s$.

Perhaps the most difficult challenge to the claim that
\eqref{eq:knowledge} is correct even in the case of mixed reasons
comes when the quantity of inadequate reasons vastly exceeds the
quantity of adequate reasons.  For example, suppose our agent has an
adequate reason $s$ (based on an assertion in some recent official
document) that one of the 20 members of the faculty of department $D$
is a logician; further, suppose she has inadequate reasons
$r_1,\dots,r_{19}$ (based on a mistaken understanding of which
specialties imply competence in logic) that the first 19 names listed
on the department $D$ faculty roster are logicians.  We might be
hard-pressed to say that our agent knows that department $D$ has a
logician on staff.

Perhaps this suggests that the agent in a case of mixed reasons can
only be said to know the proposition if she also knows that her
reasons are adequate.  We resist this move, basically because we
accept that an agent may have an adequate reason without necessarily
knowing that that reason is adequate (more on this below in the
conclusion).  Therefore, if we assume for the sake of argument that
our agent values all reasons equally, then a tiny island of adequacy
within an ocean of inadequacy is sufficient for the defender of
mixed-reason knowledge.  This grants that the agent's reasons are in
some sense confused or that an agent who has only adequate reasons
(and hence satisfies $\JTB{+}\NIL$) seems to ``know better'' than the
agent with mixed reasons.  But if one agent ``knows better,'' it does
not follow that the other does not know at
all.

In our view, an account of knowledge that would not allow for mixed
reasons would in fact be too demanding. We consider it a virtue of our
account that it allows for mixed reasons, precisely because we think
it gives us a more realistic picture of the process of acquiring and
managing reasons. We say that an agent knows a proposition if he
believes that proposition based on at least one adequate reason. But
knowing is more than passively believing propositions on the basis of
reasons. It obviously also involves comparing and relating
reasons. Consider an agent who believes the true proposition $p$ on
the basis of $r$ and $r'$, but who eventually realizes that $r'$
supports a false proposition. The agent ought to revise her belief in
$r'$, and also to reconsider her reasons for $p$. Hence, while our
account of knowledge commits us to what Dutant calls method
infallibilism, it does make room for errors and revisions in how
reasons are acquired. It contains, in that sense, a measure of
fallibilism.

\subsection{Knowledge from inadequate reasons?}

Although our account of knowledge allows for mixed reasons, a converse
objection is that it may turn out to be too strong relative to
ordinary knowledge ascriptions. The problem in this case is even more
radical, and concerns whether one can have knowledge on the basis of a
single, inadequate reason.

Here is an example:\footnote{We are indebted to Elia Zardini for
  raising the objection, and for the first example.} based on his
seeing Jones driving a Ford in the past (let us call that reason $r$),
Smith comes to wrongly believe that Jones owns a Ford. Let us modify
the scenario and suppose that Jones does in fact own a car (say, a
Mazda). Based on his seeing Jones drive a Ford in the past, Smith also
comes to believe that Jones owns a car. Could it not happen,
intuitively, that although Smith fails to have knowledge on the basis
of $r$ that Jones owns a Ford ($p$), he nevertheless has knowledge on
the basis of $r$ that Smith owns a car ($q$)? In our system, this is
not possible, for by (A), if $r$ is adequate for $q$, then $r$ must be
adequate for any other proposition that $r$ supports, hence for $p$ as
well. Dretske's treatment of conclusive reasons would be able to
address this problem: a reason $r$ can be conclusive for $q$ without
being conclusive for $p$, even when $p$ entails $q$. Our approach does
not have this flexibility.\footnote{\label{fn}One option we do not
  explore here: change adequacy from a unary property on reasons to a
  binary property on reasons and propositions. Thus instead of having
  ``$r$'' for adequacy of reason $r$ with respect to all propositions
  it supports, we would have ``$A(r,p)$'' for adequacy of $r$ with
  respect to proposition $p$. We would have to adjust other axioms:
  (A) would be $r\col\varphi\to(A(r,\varphi)\to\varphi)$, (BA) would
  be $B(r\col\varphi)\to(B(A(r,\varphi))\to B\varphi)$, and (AS) would
  be $B(r\col\varphi)\to(A(r,\varphi)\to(r\col\varphi))$. This would
  also require a change to the semantics that would make us much more
  in-line with the syntactic dependencies connecting reasons with
  specific formulas, which is familiar from Justification Logic
  \cite{ArtFit11:SEP}. This approach would add flexibility at the cost
  of simplicity.}

The question, more generally, is whether the same reason can
adequately justify one to believe $q$ without adequately justifying
one to believe a stronger proposition $p$. Similar cases have been
discussed by Warfield \cite{War05:KFF}, Fitelson
\cite{Fit10:strengthening}, and Sorensen
\cite{Sor15:fugu}. Warfield argues that I may know that I am not
late for the meeting if I believe that it is currently 2:58pm, when in
fact it is 2:56pm, assuming the meeting is at 7pm. On our account, my
reason to believe that it is currently less than 7pm is inadequate,
simply because it also supports the false proposition that it is
2:58pm. This is a case in which I have $\JTB^i$ that it is less than
7pm, without having $\JTB^e$ that it is less than 7pm. For anyone
whose intuition is that I do hold knowledge that it is is less than
7pm on the basis of my observing ``2:58pm'' on the watch, our equation
between knowledge and $\JTB^e$ is too strong in this case.

One option in the face of such examples is to bite the bullet and to
resist the intuition that I know I am not late for the meeting, or
that I know that Jones owns a car. But we think this is not the right
response. My evidence ``2:58pm'' is obviously wrong regarding the
actual time, but still close enough to the actual time to be
relevantly used. The case would be different, it seems to us, if the
agent's watch indicated 6pm when it is 2:56pm, or even 9am. For the
latter cases, our intuition is that I merely have a luckily true
belief. More generally, we think the problem concerns how much
approximation is tolerated in forming beliefs based on one's
evidence. If, when I see ``2.58pm'' ($r$) on my watch, I form the
belief ``it is around 2:58pm'' ($p$), and from that proposition I
infer ``it is less than 7pm'' ($q$), then my reason $r$ now is
veridical for both $p$ and $q$. A way out, therefore, might be to
relativize the adequacy of a reason to the selection of an appropriate
domain of propositions supported by that reason.

This nevertheless puts pressure on us to clarify the relation of
support between a reason and a proposition. In our statement of the
axiom (A), we include no restriction on the support relation. We think
it is better to be normative, and not to include any such restriction
in the definition of knowledge in terms of $\JTB^e$. On the other
hand, we are ready to accept that in actual ascriptions of knowledge,
the definition of adequate evidence is relativized to various domains
of propositions. Consider the Warfield example again: this is not a
perfect case of knowledge. This still counts as evidence that comes
close to adequate, though not perfectly adequate. But it is adequate
given the relevant domain of propositions. Our account, therefore,
does not rule out the familiar mechanisms of contextualization at play
in ordinary knowledge ascriptions, despite being fundamentally more
normative.

\section{Conclusion}

Is knowledge the same thing as JTB? We wrote this paper based on a
persistent feeling that both answers are defensible. For the negative:
Gettier's examples show that plausible reasons may be inadequate.  For
the positive: a JTB based on an adequate reason seems to rule out the
possibility of Gettier cases and can arguably be construed as a form
of knowledge.

We have shown that our framework is sufficient to address reason-based
belief and that it can be applied to important notions in
epistemology.  However, we have neglected to provide a further
analysis of ``adequacy of a reason'' into more primitive concepts.
While this notion was used as a primitive in this paper, an in-depth
study of this notion may be required in a full philosophical analysis
of the concept of knowledge.  Regardless, we think that our three-part
hierarchy of reasons (non-veridical, veridical for a proposition but
inadequate, and adequate) is itself sufficiently fruitful to
legitimate our approach. Our account, as we have explained,
fundamentally commits us to a form of infallibilism in the definition
of knowledge. But our treatment of mixed reasons also makes room for
the possibility of errors, since inadequate reasons typically coexist
with adequate reasons. We can therefore distinguish two levels in our
account of knowledge: the level of atomic reasons (and of their
support to various propositions), and the level of the network of
reasons that an agent needs to compare and manage. A discussion of
that second level lies beyond the scope of this paper, but it deserves
to be considered, because our externalist account of the notion of
adequacy remains compatible with a more internalist perspective on
knowledge.

A related question on which we propose to end is the following: how
does an agent know whether a reason is adequate?  According to
\eqref{eq:knowledge}, the agent knows $p$ if and only if there exists
an adequate reason $r$ that the agent believes is adequate and
supports $p$. Therefore, the agent knows $r$ is adequate if and only
if there exists a reason $s$ that she believes is adequate and
supports $r$ (i.e., $s\land Bs\land B(s\col r)$).\footnote{Note that
  we may have $s=r$.  In particular, it is consistent with our theory
  for reasons to be self-supporting (i.e., $r:r$).  It is also
  consistent with our theory for reasons to be non--self-supporting
  (i.e., $\lnot (r:r)$).  Since our theory permits either option, it
  is up to the user of our theory to choose which way to go as per her
  preference. We also note that it is consistent for $r$ to be
  self-rejecting (i.e., $r\col\lnot r$), and it is consistent for $r$
  to be non--self-rejecting (i.e., $\lnot (r\col\lnot r)$).}  Our
framework therefore admits the possibility that the agent may know $p$
based on an adequate reason $r$ without knowing that $r$ is itself
adequate.  In this, our framework supports the main contention of an
externalist account of knowledge: one may know $p$ without knowing
that one knows $p$ (see \cite{Dre71:CR,Wil00:Book}).\footnote{The
  internalist who objects to this need not despair: though we do not
  do so here, it is possible to extend our framework so that knowledge
  is internalizable; see \cite{ArtFit11:SEP} (and the ``proof
  checker'' operator) for details.}  We think this is right for the
externalist, though we emphasize that our theory is in principle
neutral regarding the existence of reasons justifying the adequacy of
other reasons.  

\section*{Acknowledgments}

The authors thank two anonymous reviewers of the RSL for detailed and helpful comments, as well as Wes Holliday and Sean Walsh for their editorial work. They also thank the following people for comments and discussion:
Amir Anvari,
Igor Douven,
Julien Dutant,
Rohan French,
Chenwei Shi,
Meghdad Ghari,
Andreas Herzig,
Clayton Littlejohn,
Anne Meylan,
Jennifer Nagel,
Ram Neta,
Jean-Baptiste Rauzy, 
Roy Sorensen,
Daniele Sgaravatti,
Benjamin Spector,
Giuseppe Spolaore,
Sylvia Wenmackers,
Timothy Williamson, and
Elia Zardini.

We also thank Joelle Proust, Igor Douven, Fabien Mikol, and Giuliano
Torrengo who organized the respective events where our paper was
presented (the IJN seminar on epistemic norms, the SND epistemology
colloquium, the Paris IV workshop on the Limits of Knowing, and the
Latin Meeting in Analytic Philosophy in Milan). We also thank the
various audience members present at these events for their questions.

\appendix
\section{Technical Results}

\subsection{Semantics for $\RBB_\sigma$ and $\RBB^+_\sigma$}
\label{section:sigma-semantics}

The models for $\RBB$ can be construed as models for $\RBB_\sigma$ if
we require the following additional properties:
\begin{description}
\item[(ma)] $w\in \sigma(w)$ and $r^\circ\in N(w)$ together imply that
  $w\in r(w)$.

  which says that if $\sigma$ is reflexive, then each reason $r$
  believed to be reflexive is in fact reflexive;

\item[(mb)] $\sigma^\circ\in N(w)$,

  which says that the agent believes $\sigma$ is reflexive; and

\item[(mr)] $\sem{r\col\varphi}\in N(w)$ and $r^\circ\in N(w)$
  together imply that $\sem{\sigma\col\varphi}\in N(w)$,

  which says that the agent believes $\sigma$ supports $\varphi$
  whenever she believes $r$ supports $\varphi$ and she believes $r$ is
  reflexive.
\end{description}
We write the satisfaction relation $\models_\sigma$ to indicate that
we restrict to models satisfying (ma), (mb), and (mr). By
Theorem~\ref{theorem:RBBsigma-determinacy}, $\RBB_\sigma$ is sound and
complete for the class of models satisfying (ma), (mb), and (mr).
Models for the theory $\RBB_\sigma^+$ must satisfy (ma), (mb), (mr),
and the following property:
\begin{description}
\item[(mt)] $\sem{\varphi}\in N(w)$ implies
  $\sem{\sigma\col\varphi}\in N(w)$,

  which says that the agent believes $\sigma$ supports $\varphi$
  whenever she believes $\varphi$.
\end{description}
We write $\models_\sigma^+$ to indicate that we restrict to models
satisfying (ma), (mb), (mr), and (mt).  By
Theorem~\ref{theorem:RBBsigma-determinacy}, $\RBB_\sigma^+$ is sound
and complete for the class of models satisfying (ma), (mb), (mr), and
(mt).

\subsection{Semantics for $\QRBB$, $\QRBB_\sigma$, and $\QRBB_\sigma^+$}
\label{section:QRBB-semantics}

The models for $\RBB$ can be used as models for $\QRBB$ as well.  All
that we must do is add the following satisfaction principles:
\begin{itemize}
\item $M,w\models r=s$ means that $r=s$.

\item $M,w\models(\forall r)\varphi$ means that
  $M,w\models\varphi[s/r]$ for each $s$ free for $r$ in $\varphi$.
\end{itemize}
It is shown in Theorem~\ref{theorem:QRBB-determinacy} that $\QRBB$ is
sound and complete for this semantics: for each theory we have
$\QRBB\vdash\varphi$ if and only if $\models\varphi$. However, there
is one caveat: our proof of the completeness result requires that the
set $R$ of reasons be at least countably infinite.

Additional semantic restrictions must be imposed to ensure that
$\QRBB$ models also work for $\QRBB_\sigma$ or for $\QRBB_\sigma^+$;
see \S\ref{section:sigma-semantics} for details.  Soundness and
completeness for $\QRBB_\sigma$ and for $\QRBB_\sigma^+$ follows by
Theorem~\ref{theorem:QRBBsigma-determinacy}, with the same caveat for
completeness as for $\QRBB$.

\subsection{$\RBB$ Soundness and Completeness}

We now prove the following theorem.

\begin{theorem}[$\RBB$ Soundness and Completeness]
  \label{theorem:RBB-determinacy}
  For each $\varphi\in F$:
  \[
  \RBB\vdash\varphi \quad\text{iff}\quad \models\varphi\enspace.
  \]
\end{theorem}

Soundness is by induction on the length of derivation. The arguments
for (CL), (MP), (RK), (D), (RN), and (E) are straighforward.  We
consider the remaining cases.
\begin{itemize}
\item Validity of (A): $\models r\col\varphi\to(r\to\varphi)$.
  
  $M,w\models r\col\varphi$ and $M,w\models r$ together imply that
  $r(w)\subseteq\sem{\varphi}_M$ and $w\in r(w)$.  But then
  $M,w\models\varphi$.

\item Validity of (BRK):
  $\models B(r\col (\varphi\to\psi)) \to (B(r\col \varphi)\to B(r\col\psi))$.

  Suppose $M,w\models B(r\col (\varphi\to\psi))$ and
  $M,w\models B(r\col \varphi)$. This means
  \[
  \sem{r\col(\varphi\to\psi)}_M\in N(w) \text{ and }
  \sem{r\col\varphi}_M\in N(w)\enspace.
  \]
  Applying (brk), it follows that $\sem{r\col\psi}_M\in N(w)$, which
  is what it means to have $M,w\models r\col\psi$.

\item Validity of (BA):
  $\models B(r \col \varphi) \rightarrow (Br \rightarrow B\varphi)$.

  Assume $M,w\models B(r\col\varphi)$ and $M,w\models Br$. It follows
  that $\sem{r\col\varphi}_M\in N(w)$ and $r^\circ\in N(w)$. Applying
  (ba), we obtain $\sem{\varphi}_M\in N(w)$. But this is what it means
  to have $M,w\models B\varphi$.

\item Validity of (AS):
  $\models B(r\col \varphi) \to (r \rightarrow (r\col \varphi))$.

  Assume $M,w\models B(r\col\varphi)$ and $w\in r(w)$. Hence
  $\sem{r\col\varphi}_M\in N(w)$ and $w\in r(w)$. Applying (as), we
  obtain $r(w)\subseteq\sem{\varphi}_M$. That is,
  $M,w\models r\col\varphi$.
\end{itemize}
So $\RBB$ is sound.

For completeness, we prove that $\RBB\nvdash\theta$ implies there
exists a pointed model $(M_c,\Gamma^\theta_1)$ satisfying
$M_c,\Gamma^\theta_1\not\models\theta$.  We use a canonical model
construction to build the model $M_c=(W,[\cdot],N,V)$ as follows.
First, to say that a set $S$ of formulas is \emph{consistent} means
that for no finite $S'\subseteq S$ do we have
$\RBB\vdash(\bigwedge S')\to\bot$, where $\bot$ is a fixed
contradiction such as $p\land\lnot p$.  To say a set of formulas is
\emph{maximal consistent} means that it is consistent and adding any
formula not already present will result in a set that is inconsistent
(i.e., not consistent). Let $M$ bet the set of all maximal consistent
sets of formulas.  By a standard Lindenbaum construction, it follows
that $\{\lnot\theta\}$ can be extended to some $\Gamma^\theta\in M$
and therefore $M$ is not empty.  We define $W:=M\times\{1,2\}$ and
will write $(\Gamma,i)\in W$ in the abbreviated form $\Gamma_i$. Since
$M$ is nonempty, $W$ is nonempty. For each reason $r\in R$ and
$\Gamma\in M$, define the set
\[
\Gamma^r := \{\varphi\in F \mid r\col\varphi\in\Gamma\}
\]
of $r$-supported formulas in $\Gamma$.  We then define $[r]$ by
setting
\[
r(\Gamma_i):=\{\Delta_j\in W \mid \Gamma^r\subseteq\Delta \text{ \&
} (\lnot r\in\Gamma\Rightarrow \Delta_j\neq\Gamma_i) \}\enspace.
\]
This way, a world $\Delta_j$ is $r$-accessible from $\Gamma_i$ iff
$\Delta_j$ contains all formulas $\varphi$ that are $r$-supported at
$\Gamma_i$ (as per membership of $r\col\varphi$ in $\Gamma$), unless
of course $\Delta_j=\Gamma_i$ and reflexivity is forbidden by
$\lnot r\in\Gamma$.  For each formula $\varphi\in F$, define
\[
W(\varphi):=\{\Gamma_i\in W\mid\varphi\in\Gamma\}
\]
to be the set of worlds defined by the formula $\varphi$. Then let
\[
N^+ := \{ X\subseteq W \mid \forall\varphi\in F:X\neq W(\varphi)\}
\]
be the set of worlds not definable by any formula.  For each
$\Gamma_i\in W$, we define
\[
N^+(\Gamma_i):=\{X\in N^+\mid \exists Br\in\Gamma:
r(\Gamma_i)\subseteq X \text{ \& } 
\forall\theta\in\Gamma^r(B(r\col\theta)\in\Gamma)\}\enspace.
\]
Intuitively, this is the set of non--formula-definable neighborhoods
$X$ for which there is a reason the agent accepts that supports $X$
and the agent believes this reason supports all the formulas it
actually does support.  The neighborhood function $N$ is then defined
by
\[
N(\Gamma_i) := \{X\subseteq W\mid \exists B\varphi\in\Gamma:
X=W(\varphi)\}\cup N^+(\Gamma_i)\enspace.
\]
Therefore, an agent believes a neighborhood $X$ iff the agent believes
some formula $\varphi$ that defines $X$ or, if $X$ is
non--formula-definable, there is an accepted reason supporting $X$ and
the agent believes the reason supports all the formulas it actually
does support.  Finally, we define the valuation by
\[
V(\Gamma_i):=\{p\in P\mid p\in\Gamma\}\enspace.
\]
This defines $M_c$. To check that $M_c$ is indeed a pre-model, we must
verify that $M_c$ satisfies the property (pr). We prove this now.
\begin{itemize}
\item (pr): if $x\in P\cap R$, then $x\in V(\Gamma_i)$ if and only if
  $\Gamma_i\in x(\Gamma_i)$.

  So assume $x\in V(\Gamma_i)$.  This means $x\in\Gamma$.  But then we
  have $\lnot x\notin\Gamma$ by the consistency of $\Gamma$.  Further,
  since $x\in\Gamma$, it follows by (A) and the maximal consistency of
  $\Gamma$ that $\Gamma^x\subseteq\Gamma$.  But $\lnot x\notin\Gamma$
  and $\Gamma^x\subseteq\Gamma$ together imply
  $\Gamma_i\in x(\Gamma_i)$, which completes the argument for this
  direction.  

  For the converse direction, if $\Gamma_i\in x(\Gamma_i)$, then it
  follows by the definition of $x(\Gamma_i)$ that
  $\lnot x\notin\Gamma$.  So $x\in\Gamma$ by the maximal consistency
  of $\Gamma$.  But then we have $x\in V(\Gamma_i)$ by the definition
  of $V$.
\end{itemize}
So $M_c$ is indeed a pre-model.

We prove the \emph{Truth Lemma\/}: for each formula $\varphi\in F$ and
world $\Gamma_i\in W$, we have $\varphi\in\Gamma$ iff
$M_c,\Gamma_i\models\varphi$.  The proof is by induction on the
construction of formulas, and the arguments for the base and Boolean
inductive step cases are standard, so we only consider the remaining
non-Boolean inductive step cases.
\begin{itemize}
\item Inductive step: $r\in\Gamma$ iff $M_c,\Gamma_i\models r$.
  
  If $r\in\Gamma$, then it follows by (A) and maximal consistency
  that $\Gamma^r\subseteq\Gamma$ and therefore that $\Gamma_i\in
  r(\Gamma_i)$.  But this is what it means to have
  $M_c,\Gamma_i\models r$.

  Conversely, if $M_c,\Gamma_i\models r$, then we have $\Gamma_i\in
  r(\Gamma_i)$.  By the definition of $N(\Gamma_i)$, we have $\lnot
  r\notin\Gamma$ and therefore $r\in\Gamma$ by maximal consistency.

\item Inductive step: $r\col\varphi\in\Gamma$ iff $M_c,\Gamma_i\models
  r\col\varphi$.

  If $r\col\varphi\in\Gamma$, then we have $r(\Gamma_i)\subseteq
  W(\varphi)$.  By the induction hypothesis,
  $r(\Gamma_i)\subseteq\sem{\varphi}_{M_c}$.  But this is what it means
  to have $M_c,w\models r\col\varphi$.

  Conversely, if $M_c,w\models r\col\varphi$, then we have
  $r(\Gamma_i)\subseteq\sem{\varphi}_{M_c}$ and hence
  $r(\Gamma_i)\subseteq W(\varphi)$ by the induction hypothesis.
  Toward a contradiction, assume $\lnot r\col\varphi\in\Gamma$.  It
  follows that $\Gamma^r\cup\{\lnot\varphi\}$ is consistent, for
  otherwise there would exist a finite $S\subseteq\Gamma^r$ such that
  $\vdash(\bigwedge S)\to\varphi$, hence $\vdash(\bigwedge_{\chi\in
    S}r\col\chi)\to r\col\varphi$ by modal reasoning, and therefore
  $r\col\varphi\in\Gamma$, which is impossible because it would follow
  by the assumption $\lnot r\col\varphi\in\Gamma$ that $\Gamma$ is
  inconsistent.  So we may extend the consistent set
  $\Gamma^r\cup\{\lnot\varphi\}$ to some $\Delta\in M$.  Taking $j\neq
  i$, it follows that $\Delta_j\notin W(\varphi)$ and $\Delta_j\in
  r(\Gamma_i)$, which contradicts $r(\Gamma_i)\subseteq W(\varphi)$.
  So our assumption that $\lnot r\col\varphi\in\Gamma$ is incorrect;
  what we actually have is that $\lnot r\col\varphi\notin\Gamma$ and
  therefore that $r\col\varphi\in\Gamma$ by maximal consistency.

\item Inductive step: $B\varphi\in\Gamma$ iff $M_c,\Gamma_i\models
  B\varphi$.

  If $B\varphi\in\Gamma$, then it follows by the definition of
  $N(\Gamma_i)$ that $W(\varphi)\in N(\Gamma_i)$.  By the induction
  hypothesis, $W(\varphi)=\sem{\varphi}_{M_c}$, and hence
  $\sem{\varphi}_{M_c}\in N(\Gamma_i)$.  But this is what it means to
  have $M_c,\Gamma_i\models B\varphi$.

  Conversely, if $M_c,w\models B\varphi$, then we have
  $\sem{\varphi}_{M_c}\in N(\Gamma_i)$ and therefore that
  $W(\varphi)\in N(\Gamma_i)$ by the induction hypothesis. Since
  $W(\varphi)\notin N^+$, it follows that there exists
  $B\psi\in\Gamma$ such that $W(\psi)=W(\varphi)$.  From this we have
  that $\vdash\psi\leftrightarrow\varphi$, for otherwise
  $\{\psi,\lnot\varphi\}$ could be extended to $\Delta\in M$
  satisfying $\Delta_i\in W(\psi)-W(\varphi)$ or
  $\{\lnot\psi,\varphi\}$ could be extended to $\Omega\in M$
  satisfying $\Omega_i\in W(\varphi)-W(\psi)$, but both contradict
  $W(\varphi)=W(\psi)$.  Applying (E), we have $\vdash
  B\psi\leftrightarrow B\varphi$ and therefore it follows by
  maximal consistency that $B\varphi\in\Gamma$.
\end{itemize}
This completes the proof of the Truth Lemma. 

We prove the following \emph{Consistency Lemma\/}: for each $r\in R$
and $\Gamma_i\in W$, if $Br\in\Gamma$ and
$\forall\theta\in\Gamma^r(B(r\col\theta)\in\Gamma)$, then $\Gamma^r$
is consistent.  Proceeding, assume $Br\in\Gamma$ and
$\forall\theta\in\Gamma^r(B(r\col\theta)\in\Gamma)$. Since
$Br\in\Gamma$, we have $\lnot B\lnot r\in\Gamma$ by (D) and maximal
consistency.  Toward a contradiction, suppose $\Gamma^r$ is not
consistent.  Then there exists a finite $S\subseteq\Gamma^r$ such that
$\vdash(\bigwedge S)\to\bot$.  Hence
$\vdash(\bigwedge_{\chi\in S}r\col\chi)\to r\col\bot$ by modal
reasoning. Applying maximal consistency and the fact that
$S\subseteq\Gamma^r$, we obtain $r\col\bot\in\Gamma$.  By maximal
consistency and the fact that $\vdash r\col\bot\to r\col\varphi$ for
any $\varphi$, we obtain $r\col\lnot r\in\Gamma$ and hence
$\lnot r\in\Gamma^r$.  Applying the assumption
$\forall\theta\in\Gamma^r(B(r\col\theta)\in\Gamma)$, it follows that
$B(r\col\lnot r)\in\Gamma$. Applying this and the assumption that
$Br\in\Gamma$, it follows by (BA) that $B\lnot r\in\Gamma$.  Since
$\lnot B\lnot r\in\Gamma$, it follows that the maximal consistent set
$\Gamma$ is not consistent, a contradiction.  Conclusion: $\Gamma^r$
is consistent. This completes the proof of the Consistency Lemma.

We now prove that $M_c$ is a model (and not just a pre-model); that
is, we prove that $M_c$ satisfies the properties (brk), (ba), (as),
and (d).
\begin{itemize} 
\item $M_c$ satisfies (brk): if
  $\sem{r\col(\varphi\to\psi)}_{M_c}\in N(\Gamma_i)$ and
  $\sem{r\col\varphi}_{M_c}\in N(\Gamma_i)$, then
  $\sem{r\col\psi}_{M_c}\in N(\Gamma_i)$.

  Assume $\sem{r\col(\varphi\to\psi)}_{M_c}\in N(\Gamma_i)$ and
  $\sem{r\col\varphi}_{M_c}\in N(\Gamma_i)$. By the Truth Lemma, it
  follows that $W(r\col(\varphi\to\psi))\in N(\Gamma_i)$ and
  $W(r\col\varphi)\in N(\Gamma_i)$. Since neither
  $W(r\col(\varphi\to\psi))$ nor $W(r\col\varphi)$ is a member of
  $N^+$, it follows by the definition of $N(\Gamma_i)$ that there
  exists $B\theta_1\in\Gamma$ and $B\theta_2\in\Gamma$ such that
  $W(\theta_1)=W(r\col(\varphi\to\psi))$ and
  $W(\theta_2)=W(r\col\varphi)$. As in the proof of the last inductive
  step of the Truth Lemma, it follows using (E) that
  $\vdash B\theta_1\leftrightarrow B(r\col(\varphi\to\psi))$ and that
  $\vdash B\theta_2\leftrightarrow B(r\col\varphi)$. Hence we have
  that $B(r\col(\varphi\to\psi))\in\Gamma$ and
  $B(r\col\varphi)\in\Gamma$ by maximal consistency.  Applying (BRK)
  and maximal consistency, we obtain $B(r\col\psi)\in\Gamma$. Applying
  the definition of $N(\Gamma_i)$, it follows that
  $W(r\col\psi)\in N(\Gamma_i)$.  Since we have that
  $W(r\col\psi)=\sem{r\col\psi}_{M_c}$ by the Truth Lemma, we conclude
  that $\sem{r\col\psi}_{M_c}\in N(\Gamma_i)$.

\item $M_c$ satisfies (ba): if
  $\sem{r\col\varphi}_{M_c}\in N(\Gamma_i)$ and
  $r^\circ\in N(\Gamma_i)$, then $\sem{\varphi}_{M_c}\in N(\Gamma_i)$.

  Assume $\sem{r\col\varphi}_{M_c}\in N(\Gamma_i)$ and
  $r^\circ\in N(\Gamma_i)$. Similar to our argument in the first part
  of the above proof for (brk), it follows from
  $\sem{r\col\varphi}_{M_c}\in N(\Gamma_i)$ that
  $B(r\col\varphi)\in\Gamma$. Since we have by the definition of
  $r^\circ$ that $r^\circ=W(r)$, it follows from
  $r^\circ\in N(\Gamma_i)$ that $W(r)\in N(\Gamma_i)$. Again by an
  argument similar to the first part of the above proof for (brk), we
  obtain from $W(r)\in N(\Gamma_i)$ that $Br\in\Gamma$. But then it
  follows from $B(r\col\varphi)\in\Gamma$ and $Br\in\Gamma$ by (BA)
  and maximal consistency that $B\varphi\in\Gamma$.  Applying the
  definition of $N(\Gamma_i)$, we obtain $W(\varphi)\in N(\Gamma_i)$.
  Applying the Truth Lemma, it follows that
  $\sem{\varphi}_{M_c}\in N(\Gamma_i)$.

\item $M_c$ satisfies (as): if
  $\sem{r\col\varphi}_{M_c}\in N(\Gamma_i)$ and
  $\Gamma_i\in r(\Gamma_i)$, then
  $r(\Gamma_i)\subseteq\sem{\varphi}_{M_c}$.

  Assume $\sem{r\col\varphi}_{M_c}\in N(\Gamma_i)$ and
  $\Gamma_i\in r(\Gamma_i)$. Similar to our argument in the first part
  of the above proof for (brk), it follows from
  $\sem{r\col\varphi}_{M_c}\in N(\Gamma_i)$ that
  $B(r\col\varphi)\in\Gamma$.  Now from $\Gamma_i\in r(\Gamma_i)$, it
  follows by the definition of $N(\Gamma_i)$ that
  \[
  \lnot r\in\Gamma\Rightarrow \Gamma_i\neq\Gamma_i\enspace.
  \]
  Hence $\lnot r\notin\Gamma$. But then it follows by maximal
  consistency that $r\in\Gamma$.

  So from $B(r\col\varphi)\in\Gamma$ and $r\in\Gamma$ we obtain by
  (AS) and maximal consistency that $r\col\varphi\in\Gamma$.  Hence
  $\varphi\in\Gamma^r$. To complete the argument, we wish to show that
  $r(\Gamma_i)\subseteq \sem{\varphi}_{M_c}$. So let us take an
  arbitrary $\Delta_j\in r(\Gamma_i)$ and prove that
  $\Delta_j\in\sem{\varphi}_{M_c}$. By the Truth Lemma, it suffices to
  prove that $\Delta_j\in W(\varphi)$. Proceeding, since
  $\Delta_j\in r(\Gamma_i)$, it follows by the definition of
  $r(\Gamma_i)$ that $\Gamma^r\subseteq\Delta$. But then we have by
  $\varphi\in\Gamma^r$ that $\varphi\in\Delta$ and hence
  $\Delta_j\in W(\varphi)$.

\item $M_c$ satisfies (d): if $X\in N(\Gamma_i)$, then
  $X':=W-X\notin N(\Gamma_i)$.  There are two cases to consider.

  First case for (d): assume $X\in N(\Gamma_i)-N^+(\Gamma_i)$.  It
  follows that there exists $B\varphi\in\Gamma$ such that
  $X=W(\varphi)$.  By (D) and the maximal consistency of $\Gamma$, we
  have $\lnot B\lnot\varphi\in\Gamma$.  Toward a contradiction, assume
  $X'\in N(\Gamma_i)$.  Since $X=W(\varphi)$, we have
  $X'=W(\lnot\varphi)$ by maximal consistency and therefore that
  $X'\notin N^+$.  Hence $X'\in N(\Gamma_i)-N^+(\Gamma_i)$, which
  means there exists $B\psi\in\Gamma$ such that $X'=W(\psi)$.  It
  follows that $\vdash\psi\leftrightarrow\lnot\varphi$, since
  otherwise $\{\psi,\varphi\}$ could be extended to some $\Delta\in M$
  satisfying $\Delta_j\in W(\psi)=X'$ and $\Delta_j\in W(\varphi)=X$
  or $\{\lnot\psi,\lnot\varphi\}$ could be extended to some
  $\Omega\in M$ satisfying $\Omega_k\in W-W(\psi)=X$ and
  $\Omega_k\in W-W(\varphi)=X'$, but both situations are impossible
  because $X'\cap X=\emptyset$. Hence
  $\vdash \psi\leftrightarrow\lnot\varphi$.  Applying (E), we obtain
  $\vdash B\psi\leftrightarrow B\lnot\varphi$ and therefore that
  $B\lnot\varphi\in\Gamma$, contradicting the consistency of $\Gamma$.
  Conclusion: $X'\notin N(\Gamma_i)$.

  Second case for (d): assume $X\in N^+(\Gamma_i)$.  This means there
  exists $Br\in\Gamma$ such that $r(\Gamma_i)\subseteq X$ and
  $\forall\theta\in\Gamma^r(B(r\col\theta)\in\Gamma)$.  Since
  $X\in N^+$, it follows that $X'\in N^+$ as well.  So, toward a
  contradiction, assume $X'\in N^+(\Gamma_i)$. Then we have
  $Br'\in\Gamma$ such that $r'(\Gamma_i)\subseteq X'$ and
  $\forall\theta\in\Gamma^{r'}(B(r'\col\theta)\in\Gamma)$, and hence
  $r'(\Gamma_i)\cap r(\Gamma_i)=\emptyset$.  If
  $\Gamma^r\cup\Gamma^{r'}$ were consistent, then we could extend this
  set to some $\Delta\in M$.  Taking $j\neq i$, it would follow that
  $\Gamma^r\subseteq\Delta$ and $\Gamma^{r'}\subseteq\Delta$ and that
  $\Delta_j\neq\Gamma_i$.  Hence we would have that
  $\Delta_j\in r'(\Gamma_i)\cap r(\Gamma_i)=\emptyset$, an
  impossibility.  So $\Gamma^r\cup\Gamma^{r'}$ is not consistent.
  Applying the Consistency Lemma and the fact we have
  $\{Br,Br'\}\subseteq\Gamma$ with
  $\forall\theta\in\Gamma^r(B(r\col\theta)\in\Gamma)$ and
  $\forall\theta\in\Gamma^{r'}(B(r'\col\theta)\in\Gamma)$, each of
  $\Gamma^r$ and $\Gamma^{r'}$ is consistent, so it follows from the
  inconsistency of $\Gamma^r\cup\Gamma^{r'}$ that there exists a
  finite nonempty subset $S$ of one of the two sets $\Gamma^r$ and
  $\Gamma^{r'}$ such that for some formula $\varphi$ that is a member
  of the other of these two sets we have
  $\vdash(\bigwedge S)\to\lnot\varphi$. Let us assume
  $S\subseteq\Gamma^r$ and $\varphi\in\Gamma^{r'}$; the argument for
  the other possibility, where $S\subseteq\Gamma^{r'}$ and
  $\varphi\in\Gamma^r$, is similar.  Proceeding, we have
  $\vdash(\bigwedge_{\chi\in S}r\col\chi)\to r\col\lnot\varphi$ by
  modal reasoning.  Since $S\subseteq\Gamma^r$, it follows by maximal
  consistency that $r\col\lnot\varphi\in\Gamma$ and hence
  $\lnot\varphi\in\Gamma^r$. As we have
  $\forall\theta\in\Gamma^r(B(r\col\theta)\in\Gamma)$, it follows that
  $B(r\col\lnot\varphi)\in\Gamma$. Since we also have that
  $Br\in\Gamma$, it follows by (BA) and maximal consistency that
  $B\lnot\varphi\in\Gamma$. But $\varphi\in\Gamma^{r'}$, from which it
  follows by $\forall\theta\in\Gamma^{r'}(B(r'\col\theta)\in\Gamma)$
  that $B(r'\col\varphi)\in\Gamma$.  Since we also have that
  $Br'\in\Gamma$, it follows by (BA) and maximal consistency that
  $B\varphi\in\Gamma$ and therefore that
  $\lnot B\lnot\varphi\in\Gamma$ by (D) and maximal consistency.  But
  then we have shown that $\lnot B\lnot\varphi\in\Gamma$ and
  $B\lnot\varphi\in\Gamma$, which implies $\Gamma$ is inconsistent, a
  contradiction.  Conclusion: $X'\notin N^+(\Gamma_i)$.
\end{itemize}
So $M_c$ is indeed a model, and therefore $(M_c,\Gamma^\theta_1)$ is a
pointed model. Thus, since $\lnot\theta\in\Gamma^\theta$, it follows
by the Truth Lemma that $M_c,\Gamma^\theta_1\not\models\theta$.
Completeness follows.

\subsection{$\RBB_\sigma$ and $\RBB_\sigma^+$ Soundness and
  Completeness}

Recalling the semantics for $\RBB_\sigma$ and for $\RBB_\sigma^+$ from
\S\ref{section:sigma-semantics}, we prove the following theorem.

\begin{theorem}[$\RBB_\sigma$ and $\RBB_\sigma^+$ Soundness and
  Completeness]
  \label{theorem:RBBsigma-determinacy}
  Assume $R$ contains the symbol $\sigma$.  For each $\varphi\in F$:
  \begin{eqnarray*}
    \RBB_\sigma\vdash\varphi &\text{iff}& \models_\sigma\varphi 
    \quad\text{and}
    \\
    \RBB_\sigma^+\vdash\varphi &\text{iff}& \models_\sigma^+\varphi 
    \quad\text{.}
  \end{eqnarray*}
\end{theorem}

Soundness for $\RBB_\sigma$ and for $\RBB_\sigma^+$ are as for $\RBB$
(Theorem~\ref{theorem:RBB-determinacy}) except that we must check
soundness of the additional axioms.  We consider each in turn.
\begin{itemize}
\item Validity of (MA): $\models_\sigma \sigma\to(Br\to r)$ and
  $\models_\sigma^+ \sigma\to(Br\to r)$.

  Assume $M,w\models_\sigma \sigma$ and $M,w\models_\sigma Br$.  This
  means $w\in\sigma(w)$ and $r^\circ\in N(w)$, from which it follows
  by (ma) that $w\in r(w)$.  But then $M,w\models_\sigma r$.  The
  argument for the satisfaction operator $\models_\sigma^+$ is the
  same.

\item Validity of (MB): $\models_\sigma B\sigma$ and $\models_\sigma^+
  B\sigma$.

  Given $(M,w)$, we have $\sigma^\circ\in N(w)$ by (mb). So
  $M,w\models_\sigma B\sigma$ and $M,w\models_\sigma^+ B\sigma$.

\item Validity of (MR):
  \begin{tabular}[t]{l}
    $\models_\sigma B(r\col\varphi)\to(Br\to B(\sigma\col\varphi))$ and
    \\
    $\models_\sigma^+ B(r\col\varphi)\to(Br\to B(\sigma\col\varphi))$.
  \end{tabular}

  Assume $M,w\models_\sigma B(r\col\varphi)$ and
  $M,w\models_\sigma Br$.  Then $\sem{r\col\varphi}_M\in N(w)$ and
  $r^\circ\in N(w)$.  Applying (mr), it follows that
  $\sem{\sigma\col\varphi}_M\in N(w)$.  But then
  $M,w\models_\sigma B(\sigma\col\varphi)$.  The argument for the
  satisfaction operator $\models_\sigma^+$ is the same.

\item Validity of (MT):
  $\models_\sigma^+ B\varphi\to B(\sigma\col\varphi)$.

  Assume $M,w\models_\sigma^+ B\varphi$.  This means
  $\sem{\varphi}_M\in N(w)$.  Applying (mt), we have
  $\sem{\sigma\col\varphi}_M\in N(w)$.  But this is what it means to
  have $M,w\models_\sigma^+ B(\sigma\col\varphi)$.
\end{itemize}
So $\RBB_\sigma$ and $\RBB_\sigma^+$ are sound.  Completeness for
$\RBB_\sigma$ and for $\RBB_\sigma^+$ is as for $\RBB$
(Theorem~\ref{theorem:RBB-determinacy}) except that provability is
taken with respect to the theory in question and we must show that
$M_c$ satisfies the additional properties required of models of this
theory (\S\ref{section:sigma-semantics}).
\begin{itemize}
\item $M_c$ satisfies (ma) for $\RBB_\sigma$ and for $\RBB_\sigma^+$:
  $\Gamma_i\in\sigma(\Gamma_i)$ and $r^\circ\in N(\Gamma_i)$ together
  imply that $\Gamma_i\in r(\Gamma_i)$.

  Assume $\Gamma_i\in\sigma(\Gamma_i)$ and $r^\circ\in N(w)$.  As in
  the proof that $M_c$ satisfies (ba) from
  Theorem~\ref{theorem:RBB-determinacy}, it follows from
  $r^\circ\in N(\Gamma_i)$ that $Br\in\Gamma$.  Applying the
  definition of $\sigma(\Gamma_i)$ to our assumption
  $\Gamma_i\in\sigma(\Gamma_i)$, it follows that
  $\lnot\sigma\notin\Gamma$ and therefore $\sigma\in\Gamma$ by maximal
  consistency.  Since $\sigma\in\Gamma$ and $Br\in\Gamma$, we have by
  (MA) and maximal consistency that $r\in\Gamma$.  But then
  $\Gamma^r\subseteq\Gamma$ by (A) and maximal consistency. Since it
  follows from $r\in\Gamma$ by maximal consistency that
  $\lnot r\notin\Gamma$ and we have shown that
  $\Gamma^r\subseteq\Gamma$, it follows by the definition of
  $r(\Gamma_i)$ that $\Gamma_i\in r(\Gamma_i)$.

\item $M_c$ satisfies (mb) for $\RBB_\sigma$ and for $\RBB_\sigma^+$:
  $\sigma^\circ\in N(\Gamma_i)$.

  We have $B\sigma\in\Gamma$ by (MB).  Hence
  $\sigma^\circ=W(\sigma)\in N(\Gamma_i)$.

\item $M_c$ satisfies (mr) for $\RBB_\sigma$ and for $\RBB_\sigma^+$:
  $\sem{r\col\varphi}_{M_c}\in N(\Gamma_i)$ and
  $r^\circ\in N(\Gamma_i)$ together imply that
  $\sem{\sigma\col\varphi}_{M_c}\in N(\Gamma_i)$.

  Assume $\sem{r\col\varphi}_{M_c}\in N(\Gamma_i)$ and
  $r^\circ\in N(\Gamma_i)$.  As in the proof that $M_c$ satisfies (ba)
  from Theorem~\ref{theorem:RBB-determinacy}, it follows that
  $B(r\col\varphi)\in\Gamma$ and $Br\in\Gamma$. Applying (MR) and
  maximal consistency, we obtain $B(\sigma\col\varphi)\in\Gamma$.  By
  the definition of $N(\Gamma_i)$, it follows that
  $W(\sigma\col\varphi)\in N(\Gamma_i)$. Applying the Truth Lemma, we
  obtain $\sem{\sigma\col\varphi}_{M_c}\in N(\Gamma_i)$.

\item $M_c$ satisfies (mt) for $\RBB_\sigma^+$:
  $\sem{\varphi}_{M_c}\in N(\Gamma_i)$ implies
  $\sem{\sigma\col\varphi}_{M_c}\in N(\Gamma_i)$.

  Assume $\sem{\varphi}_{M_c}\in N(\Gamma_i)$. Applying the Truth
  Lemma, we obtain $W(\varphi)\in N(\Gamma_i)$. We therefore have
  $B\varphi\in\Gamma$ by the definition of $N(\Gamma_i)$.  Applying
  (MT) and maximal consistency, it follows from $B\varphi\in\Gamma$
  that $B(\sigma\col\varphi)\in\Gamma$. Applying the definition of
  $N(\Gamma_i)$, we obtain $W(\sigma\col\varphi)\in N(\Gamma_i)$. We
  therefore conclude that
  $\sem{\sigma\col\varphi}_{M_c}\in N(\Gamma_i)$ by the Truth Lemma.
\end{itemize}

\subsection{Lemmas for $\QRBB$ Completeness}

The results of this section will be used in
\S\ref{section:QRBB-completeness} to prove completeness for $\QRBB$.
All provability in this section is taken with respect to $\QRBB$.

\begin{lemma}
  \label{lemma:rules}
  $\QRBB$ satisfies the following:
  \begin{itemize}
  \item \emph{Distributivity\/}: $\vdash(\forall
    r)(\varphi\to\psi)\to ((\forall r)\varphi\to(\forall r)\psi)$;

  \item the \emph{Distribution Rule\/}: $\vdash\varphi\to\psi$ implies
    $\vdash(\forall r)\varphi\to(\forall r)\psi$;

  \item the \emph{Renaming Rule\/}: if $s$ has no occurrence in
    $\varphi$, then 
    \begin{center}
      $\vdash (\forall r)\varphi\leftrightarrow(\forall s)
      \varphi[s/r]$ \quad{}and\quad
      $\vdash (\exists r)\varphi\leftrightarrow(\exists s)
      \varphi[s/r]$ ;
    \end{center}

  \item the \emph{Equivalence Rule\/}:
    $\vdash\varphi\leftrightarrow\varphi'$ implies
    \begin{center}
      $\vdash(\forall r)\varphi\leftrightarrow(\forall r)\varphi'$
      \quad{}and\quad $\vdash(\exists r)\varphi\leftrightarrow(\exists
      r)\varphi'$ .
    \end{center}
  \end{itemize}
\end{lemma}
\begin{proof}
  For Distributivity:
  \begin{center}
  \renewcommand{\arraystretch}{1.3}\begin{tabular}{lll} 
    1. &
    $\vdash (\forall r)(\varphi\to\psi)\to(\varphi\to\psi)$
    & (UI)
    \\
    2. &
    $\vdash (\forall r)\varphi\to\varphi$
    & (UI)
    \\
    3. & 
    $\vdash 
    ((\forall r)(\varphi\to\psi)\land(\forall r)\varphi)\to\psi$
    & by 1, 2
    \\
    4. &
    $\vdash (\forall r)
    [((\forall r)(\varphi\to\psi)\land(\forall r)\varphi)\to\psi]$
    & by 3, (Gen)
    \\
    5. &
    $\vdash ((\forall r)(\varphi\to\psi)\land(\forall r)\varphi)\to
    (\forall r)\psi$
    & by 4, (UD), (MP)
    \\
    6. &
    $\vdash(\forall
    r)(\varphi\to\psi)\to ((\forall r)\varphi\to(\forall r)\psi)$
    & by 5
  \end{tabular}
  \end{center}

  For the Distribution Rule: 
  \begin{center}
  \renewcommand{\arraystretch}{1.3}\begin{tabular}{lll} 
    1. &
    $\vdash\varphi\to\psi$
    & assumption
    \\
    2. &
    $\vdash(\forall r)\varphi\to\varphi$
    & (UI)
    \\
    3. &
    $\vdash(\forall r)\varphi\to\psi$
    & by 1, 2
    \\
    4. &
    $\vdash(\forall r)\varphi\to(\forall r)\psi$
    & by 3, (UD), (MP)
  \end{tabular}
  \end{center}

  For the Renaming Rule: if $s$ has no occurrence in $\varphi$, then
  \begin{center}
  \renewcommand{\arraystretch}{1.3}\begin{tabular}{lll} 
    1. &
    $\vdash(\forall r)\varphi\to\varphi[s/r]$
    & (UI)
    \\
    2. &
    $\vdash(\forall s)((\forall r)\varphi\to\varphi[s/r])$
    & by 1, (Gen)
    \\
    3. &
    $\vdash(\forall r)\varphi\to(\forall s)\varphi[s/r]$
    & by 2, (UD), (MP)
    \\
    4. &
    $\vdash(\forall s)\varphi[s/r]\to\varphi[s/r][r/s]$
    & (UI)
    \\
    5. &
    $\vdash(\forall s)\varphi[s/r]\to\varphi$
    & by 4, $\varphi[s/r][r/s]=\varphi$
    \\
    6. &
    $\vdash(\forall r)((\forall s)\varphi[s/r]\to\varphi)$
    & by 5, (Gen)
    \\
    7. &
    $\vdash(\forall s)\varphi[s/r]\to(\forall r)\varphi$
    & by 6, (UD), (MP)
    \\
    8. &
    $\vdash(\forall r)\varphi\leftrightarrow(\forall s)\varphi[s/r]$
    & by 3, 7 --- our first result
    \\
    9. &
    $\vdash(\forall r)\lnot\varphi\leftrightarrow(\forall s)(\lnot\varphi)[s/r]$
    & by an argument like 1--8
    \\
    10. &
    $\vdash(\forall r)\lnot\varphi\leftrightarrow(\forall s)\lnot\varphi[s/r]$
    & by 9, $(\lnot\varphi)[s/r]=\lnot(\varphi[s/r])$
    \\
    11. &
    $\vdash\lnot(\forall r)\lnot\varphi\leftrightarrow
          \lnot(\forall s)\lnot\varphi[s/r]$
    & by 10
    \\
    12. &
    $\vdash(\exists r)\varphi\leftrightarrow(\exists s)\varphi[s/r]$
    & by 11
  \end{tabular}
  \end{center}

  For the Equivalence Rule: 
  \begin{center}
  \renewcommand{\arraystretch}{1.3}\begin{tabular}{lll} 
    1. &
    $\vdash\varphi\leftrightarrow\varphi'$
    & assumption
    \\
    2. &
    $\vdash\varphi\to\varphi'$
    & by 1
    \\
    3. &
    $\vdash(\forall r)\varphi\to(\forall r)\varphi'$
    & by 2, Distribution Rule
    \\
    4. &
    $\vdash\varphi'\to\varphi$
    & by 1
    \\
    5. &
    $\vdash(\forall r)\varphi'\to(\forall r)\varphi$
    & by 4, Distribution Rule
    \\
    6. &
    $\vdash(\forall r)\varphi\leftrightarrow(\forall r)\varphi'$
    & by 3, 5 --- our first result
    \\
    7. &
    $\vdash\lnot\varphi\leftrightarrow\lnot\varphi'$
    & by 1
    \\
    8. &
    $\vdash 
    (\forall r)\lnot\varphi\leftrightarrow
    (\forall r)\lnot\varphi'$
    & from 7 by an argument like 1--6
    \\
    9. &
    $\vdash
    \lnot(\forall r)\lnot\varphi\leftrightarrow
    \lnot(\forall r)\lnot\varphi'$
    & by 8
    \\
    10. &
    $\vdash
    (\exists r)\varphi\leftrightarrow(\exists r)\varphi'$
    & by 9
  \end{tabular}
  \end{center}
\end{proof}

\begin{lemma}
  \label{lemma:quantifiers}
  If $r$ is not free in $\psi$, then:
  \begin{enumerate}
  \item \label{lemma:quantifiers:a} $\vdash (\forall
    r)(\varphi\to\psi)\to ((\exists r)\varphi\to\psi)$; and

  \item \label{lemma:quantifiers:b} $\vdash (\exists
    r)(\psi\to\varphi)\to (\psi\to(\exists r)\varphi)$.
  \end{enumerate}
\end{lemma}
\begin{proof}
  For Item~\ref{lemma:quantifiers:a}: if $r$ is not free in $\psi$,
  then
  \begin{center}
  \renewcommand{\arraystretch}{1.3}\begin{tabular}{lll} 
    1. &
    $\vdash(\forall r)(\varphi\to\psi)\to(\forall
    r)(\lnot\psi\to\lnot\varphi)$ & \text{Equivalence Rule, Classical Logic}
    \\
    2. & $\vdash(\forall
    r)(\lnot\psi\to\lnot\varphi)\to(\lnot\psi\to(\forall
    r)\lnot\varphi)$ & \text{(UI)}
    \\
    3. & $\vdash(\forall r)(\varphi\to\psi)\to (\lnot\psi\to(\forall
    r)\lnot\varphi)$ & \text{by 1, 2}
    \\
    4. & $\vdash(\forall r)(\varphi\to\psi)\to (\lnot(\forall
    r)\lnot\varphi\to\psi)$ & \text{from 3}
    \\
    5. & $\vdash (\forall r)(\varphi\to\psi)\to ((\exists
    r)\varphi\to\psi)$ & \text{from 4}
  \end{tabular}
  \end{center}
  For Item~\ref{lemma:quantifiers:b}: if $r$ is not free in $\psi$,
  then
  \begin{center}
  \renewcommand{\arraystretch}{1.3}\begin{tabular}{lll}
    1. &
    $\vdash 
    \psi\to(\lnot\varphi\to(\psi\land\lnot\varphi))$
    & (CL)
    \\
    2. &
    $\vdash
    (\forall r)\psi \to
    (\forall r)(\lnot\varphi\to(\psi\land\lnot\varphi))$
    & by 1, Distribution Rule
    \\
    3. & 
    $\vdash
    (\forall r)\psi \to
    ((\forall r)\lnot\varphi\to(\forall r)(\psi\land\lnot\varphi))$
    & by 2, Distributivity
    \\
    4. &
    $\vdash
    ((\forall r)\psi\land(\forall r)\lnot\varphi) \to
    (\forall r)(\psi\land\lnot\varphi)$
    & by 3
    \\
    5. &
    $\vdash\psi\to\psi$
    & (CL)
    \\
    6. &
    $\vdash(\forall r)(\psi\to\psi)$
    & by 5, (Gen)
    \\
    7. &
    $\vdash\psi\to(\forall r)\psi$
    & by 6, (UD), (MP)
    \\
    8. &
    $\vdash
    (\psi\land(\forall r)\lnot\varphi) \to
    (\forall r)(\psi\land\lnot\varphi)$
    & by 4, 7
    \\
    9. &
    $\vdash
    \lnot(\forall r)(\psi\land\lnot\varphi) \to
    (\psi\to\lnot(\forall r)\lnot\varphi)$
    & by 8
    \\
    10. &
    $\vdash
    \lnot(\forall r)\lnot(\psi\to\varphi) \to
    (\psi\to\lnot(\forall r)\lnot\varphi)$
    & by 9, Equivalence Rule
    \\
    11. &
    $\vdash (\exists
    r)(\psi\to\varphi)\to (\psi\to(\exists r)\varphi)$
    & by 10
  \end{tabular}
  \end{center}
\end{proof}

\subsection{$\QRBB$ Soundness and Completeness}
\label{section:QRBB-completeness}

Recalling the semantics for $\QRBB$ from
\S\ref{section:QRBB-semantics}, we prove the following theorem.

\begin{theorem}[$\QRBB$ Soundness and Completeness]
  \label{theorem:QRBB-determinacy}
  We have:
  \begin{itemize}
  \item $\QRBB$ is sound: $\QRBB\vdash\varphi$ implies
    $\models\varphi$ for each $\varphi\in F^\forall$; and

  \item if $R$ is at least countably infinite, then $\QRBB$ is sound
    and complete: for each $\varphi\in F^\forall$,
    \[
    \QRBB\vdash\varphi \quad\text{iff}\quad \models\varphi\enspace.
    \]
  \end{itemize}
\end{theorem}

Soundness is by induction on the length of derivation.  Most cases are
addressed in the proof of Theorem~\ref{theorem:RBB-determinacy}.  We
only address the remaining cases.
\begin{itemize}
\item Validity of (UD): $\models (\forall r)(\varphi\to\psi)\to
  (\varphi\to(\forall r)\psi)$, where $r$ is not free in $\varphi$.
  
  Assume $M,w\models(\forall r)(\varphi\to\psi)$ and
  $M,w\models\varphi$. From the former, we have
  $M,w\models(\varphi\to\psi)[s/r]$ for each $s$ free for $r$ in
  $\varphi\to\psi$.  Since $r$ is not free in $\varphi$, it follows
  that $M,w\models\varphi\to\psi[s/r]$ for each $s$ free for $r$ in
  $\psi$.  By our assumption $M,w\models\varphi$, it follows that
  $M,w\models\psi[s/r]$ for each $s$ free for $r$ in $\psi$.  That is,
  $M,w\models(\forall r)\psi$.

\item Validity of (UI): $\models (\forall r)\varphi\to\varphi[s/r]$,
  where $s$ is free for $r$ in $\varphi$.

  By the definition of satisfaction.

\item Validity of (EP) and (EN): $\models r=r$ and
  $\models\lnot(r=s)$, where $r$ and $s$ are different.

  By the definition of satisfaction.

\item (Gen) preserves validity: $\models\varphi$ implies
  $\models(\forall r)\varphi$.

  If $\not\models(\forall r)\varphi$, then there exists $(M,w)$ and
  $s$ free for $r$ in $\varphi$ such that
  $M,w\not\models\varphi[s/r]$.  Given that
  $M=(W,[\cdot],N,V)$, define the model
  $M'=(W,[\cdot]',N,V)$ by setting
  \[
  [t]':=
  \begin{cases}
    [s] & \text{if } t=r, \\
    [t] & \text{otherwise}.
  \end{cases}
  \]
  It follows that $r^{M'}(w)=s^M(w)$ and $t^{M'}(w)=t^M(w)$ for all
  $t\neq r$.  By the usual arguments about the preservation of truth
  of formulas under the renaming of quantified variables and their
  corresponding bound occurrences, we may assume without loss of
  generality that every occurrence of $r$ in $\varphi$ is free. But
  then it is easy to see that we have $M,w\not\models\varphi[s/r]$ iff
  $M',w\not\models\varphi$.  After all, $\varphi$ and $\varphi[s/r]$
  differ only in the occurrences of $r$ in $\varphi$ that are replaced
  by $s$ in $\varphi[s/r]$, and $M'$ interprets all such occurrences
  of $r$ just as $M$ interprets the corresponding occurrences of $s$.
  And all other occurrences of symbols in $\varphi$ are the same as
  they are in $\varphi[s/r]$, they are syntactically different than
  $r$, and $M'$ and $M$ interpret them in the same way.  Conclusion:
  $\not\models\varphi$.
\end{itemize}
So $\QRBB$ is sound.

For completeness, we adapt the standard Henkin-style construction in
\cite[\S3.1]{vDal94:Book} to the present setting. To begin, our
language $F^\forall$ depends on two parameters: a nonempty set $R$ of
reasons and a nonempty set $P$ of propositional letters.  We shall
keep $P$ fixed but consider different options for $R$.  As such, it
will be convenient to write $L(R)$ to denote the set of formulas with
quantifiers that we can form using $R\neq\emptyset$ as our set of
reasons.  By convention in this proof, we restrict all derivation to
be with respect to $\QRBB$.  Also, we shall assume for the remainder
of the argument that $R$ is at least countably infinite.

To say that a set $\Gamma\subseteq L(R)$ is \emph{consistent} means
that for no finite $\Gamma'\subseteq\Gamma$ is it the case that
$\vdash(\bigwedge \Gamma')\to\bot$.  To say that $\Gamma\subseteq
L(R)$ is \emph{maximal $L(R)$-consistent} means that $\Gamma$ is
consistent and adding to $\Gamma$ any formula of $L(R)$ not already
present would produce an inconsistent set.

For the purposes of the present proof, a \emph{theory} in the language
$L(R)$ is a set $T\subseteq L(R)$ of formulas in $L(R)$ satisfying the
following properties:
\begin{itemize}
\item \emph{Closure under theorems\/}: if $\varphi\in L(R)$ and
  $\vdash\varphi$, then $\varphi\in T$; and

\item \emph{Closure under} (MP): if $\varphi\to\psi\in T$ and
  $\varphi\in T$, then $\psi\in T$.
\end{itemize}

Given $\Gamma\subseteq L(R)$, let $\mathbb{T}_R(\Gamma)$ be the set of
all theories in $L(R)$ that contain $\Gamma$.  The intersection of a
collection of theories in $L(R)$ is also a theory in $L(R)$.  Hence
for each $\Gamma\subseteq L(R)$, we may define the theory
\[
T_R(\Gamma) := \bigcap\mathbb{T}_R(\Gamma)
\]
called the \emph{theory in $L(R)$ generated by $\Gamma$}.

An \emph{$L(R)$-proof from $\Gamma$} is a finite nonempty sequence
$\may{\varphi_1,\dots,\varphi_n}$ of formulas in $L(R)$ such that for
each $\varphi_i$ in the sequence, we have one of the following:
$\varphi_i\in\Gamma$, $\QRBB\vdash\varphi_i$, or there exist
$\varphi_j$ and $\varphi_k$ from earlier in the sequence (i.e., $j<i$
and $k<i$) such that $\varphi_i$ follows by (MP) from $\varphi_j$ and
$\varphi_k$ (i.e., $\varphi_k=\varphi_j\to\varphi_i$).  To say that an
$L(R)$-proof from $\Gamma$ is \emph{of $\varphi$}, means that
$\varphi$ is the last formula in the sequence.  We write
$\Gamma\vdash_R\varphi$ to mean that there exists an $L(R)$-proof of
$\varphi$ from $\Gamma$. Notation: in writing the set to the left of
the turnstile $\vdash_R$, we may use a comma to denote set-theoretic
union, we may identify an individual formula with the singleton set
containing the formula in question, and we may omit any set-indicating
notation if the set is empty.  We state without proof the following
results, grouped together under the name \emph{Simple Lemma\/}:
\begin{itemize}
\item $\Gamma\vdash_R\varphi$ iff $\varphi\in T_R(\Gamma)$;

\item $\Gamma\vdash_R\varphi$ iff $T_R(\Gamma)\vdash_R\varphi$;

\item $\Gamma\vdash_R\varphi$ iff there exists a
  $\Gamma'\subseteq\Gamma$ such that $\Gamma'\vdash_R\varphi$;

\item $\Gamma\vdash_R\varphi$ iff there exists a finite
  $\Gamma'\subseteq\Gamma$ such that $\Gamma'\vdash_R\varphi$;

\item if $\Gamma$ is finite, then $\Gamma'\vdash_R\varphi$ iff
  $\bigwedge\Gamma'\vdash_R\varphi$, where
  $\bigwedge\Gamma':=\bigwedge_{\chi\in\Gamma'}\chi$; and

\item $\vdash_R\varphi$ iff $\QRBB\vdash\varphi$.
\end{itemize}
Generally the Simple Lemma will be used only tacitly.

Given a theory $T$ in $L(R)$ and a theory $T'$ in $L(R')$, to say that
$T'$ is an \emph{extension} of $T$ means that $T\subseteq T'$.  To say
that $T'$ is a \emph{conservative extension} of $T$ means that $T'\cap
L(R)=T$.

To say that a theory $T$ in $L(R)$ is \emph{Henkin} means that for
each closed formula (i.e., containing no free variables) of the form
$\lnot(\forall r)\varphi\in L(R)$, there exists a reason
$r_\varphi\in R$ called a \emph{witness (or Henkin constant) for
  $\lnot(\forall r)\varphi$} for which we have
\[
(\lnot(\forall r)\varphi \to \lnot\varphi[r_\varphi/r]) \in T\enspace.
\]
Given a theory $T$ in $L(R)$, let $R^*$ be the set obtained from $R$
by adding for each closed $\lnot(\forall r)\varphi\in L(R)$ a new
reason $r_\varphi$.  To be clear: there is a bijection between the set
of closed formulas $\lnot(\forall r)\varphi\in L(R)$ and the set
$R^*-R$.  We define the set
\[
H(R):= \{\lnot(\forall r)\varphi\to\lnot\varphi[r_\varphi/r] \mid
\lnot(\forall r)\varphi\in L(R) \text{ is closed }\}
\]
of \emph{Henkin axioms in $L(R)$} and let $T^*:=T_{R^*}(T\cup H(R))$
be the theory in $L(R^*)$ generated by $T\cup H(R)$.

\begin{lemma}[Constants]
  \label{lemma:constants}
  Assume $R$ is at least countably infinite and $R\subseteq R'$.  If
  $\Gamma\cup\{\varphi\}\subseteq L(R)$, then
  $\Gamma\vdash_{R'}\varphi$ iff $\Gamma\vdash_R\varphi$.
\end{lemma}
\begin{proof}
  The right-to-left direction is immediate (since $R\subseteq R'$), so
  we prove only the left-to-right direction.  Proceeding, assume
  $\Gamma\cup\{\varphi\}\subseteq L(R)$ and
  $\Gamma\vdash_{R'}\varphi$; that is, there exists an $L(R')$-proof
  $\pi'=\may{\psi_1',\dots,\psi_n'}$ of $\varphi$ from $\Gamma$.  Let
  $r_1',\dots,r_m'$ be a non-repeating list of all reasons in $R'-R$
  that appear in $\pi'$.  Since $R$ is at least countably infinite and
  $\pi'$ is finite, we may choose a non-repeating list $r_1,\dots,r_m$
  of reasons in $R$ that do not appear anywhere in $\pi'$.  Such a
  list exists because $R$ is at least countably infinite. Form
  $\pi:=\may{\psi_1,\dots,\psi_n}$ by defining $\psi_i$ as the formula
  obtained from $\psi_i'$ by replacing all occurrences of
  $r_1',\dots,r_m'$ by $r_1,\dots,r_m$ (respectively).  Since
  $\Gamma\cup\{\varphi\}\subseteq L(R)$, one may verify that $\pi$ is
  an $L(R)$-proof of $\varphi$ from $\Gamma$ (i.e., formulas in
  $\Gamma\subseteq L(R)$ are left unchanged, $\QRBB$-theorems in
  $L(R')$ are mapped to $\QRBB$-theorems in $L(R)$, formulas in
  $L(R')$ obtained via (MP) in $\pi'$ are mapped to formulas in $L(R)$
  obtained via (MP) in $\pi$, and $\varphi\in L(R)$ is left
  unchanged).  Hence $\Gamma\vdash_R\varphi$.
\end{proof}

\begin{lemma}[Deduction]
  \label{lemma:deduction}
  For each $R$, we have:
  \[
  \Gamma\cup\{\varphi\}\vdash_R\psi
  \quad\text{iff}\quad
  \Gamma\vdash_R\varphi\to\psi\enspace.
  \]
\end{lemma}
\begin{proof}
  The right-to-left direction is easy, so we only address the
  left-to-right direction. Proceeding, assume
  $\Gamma\cup\{\varphi\}\vdash_R\psi$, which implies there exists an
  $L(R)$-proof $\may{\chi_1,\dots,\chi_n}$ of $\psi$ from
  $\Gamma\cup\{\varphi\}$. It suffices for us to prove by induction on
  $i\leq n$ that $\Gamma\vdash_R\varphi\to\chi_i$.

  In the base case, $i=1$ and therefore either $\vdash\chi_i$ or
  $\chi_i\in\Gamma\cup\{\varphi\}$.  If $\vdash\chi_i$, then
  $\vdash\varphi\to\chi_i$ and therefore
  $\Gamma\vdash_R\varphi\to\chi_i$.  If
  $\chi_i\in\Gamma\cup\{\varphi\}$, then either $\chi_i\in\Gamma$ or
  $\chi_i=\varphi$.  If $\chi_i=\varphi$, then since
  $\vdash\varphi\to\varphi$ by (CL), we have
  $\Gamma\vdash_R\varphi\to\varphi$.  So suppose $\chi_i\in\Gamma$.
  Hence $\Gamma\vdash_R\chi_i$. Since for
  $\epsilon_i:=\chi_i\to(\varphi\to\chi_i)$ we have $\vdash\epsilon_i$
  by (CL), we have $\Gamma\vdash_R\epsilon_i$ and therefore
  $\Gamma\vdash_R\varphi\to\chi_i$.

  For the induction step ($i>1$), we have that $\vdash\chi_i$, that
  $\chi_i\in\Gamma\cup\{\varphi\}$, or that $\chi_i$ follows by (MP)
  from $\chi_k$ and $\chi_k\to\chi_i$ appearing earlier in the
  $L(R)$-proof.  The argument for the first two possibilities is as in
  the base case.  So assume the third possibility obtains.  By the
  induction hypothesis, we have $\Gamma\vdash_R\varphi\to\chi_k$ and
  $\Gamma\vdash_R\varphi\to(\chi_k\to\chi_i)$ Let $\theta_i$ be the
  classical tautology
  \[
  \theta_i:= (\varphi\to(\chi_k\to\chi_i))\to
  ((\varphi\to\chi_k)\to(\varphi\to\chi_i))\enspace.
  \]
  We have $\vdash\theta_i$ by (CL) and hence $\Gamma\vdash_R\theta_i$.
  But then $\Gamma\vdash_R\varphi\to\chi_i$ by (MP).
\end{proof}

We remark that the version of Lemma~\ref{lemma:deduction} for the
$\QRBB$ consequence relation $\vdash$ does not hold in general.  For
example, we have $r\col\varphi\vdash(\forall r)(r\col\varphi)$ and yet
$\nvdash r\col\varphi\to(\forall r)(r\col\varphi)$.  As another
example, we have $p\vdash r\col p$ and yet $\nvdash p\to r\col p$.
Lemma~\ref{lemma:deduction} does not fail in similar ways because the
consequence relation given by the $R$-specific turnstile $\vdash_R$
gives rise to a notion of proof (i.e., the $L(R)$-proof) that forbids
the direct use of any $\QRBB$ rule other than (MP).

\begin{lemma}[Fresh Variable]
  \label{lemma:fresh-variable}
  If $s\in R$ does not occur in any formula in
  $\Gamma\cup\{\varphi\}\subseteq L(R)$, then, letting
  $\Gamma[s/r]:=\{\chi[s/r]\mid \chi\in\Gamma\}$, we have:
  \begin{enumerate}
  \item \label{lemma:fresh-variable:a} $\vdash\varphi$ iff
    $\vdash\varphi[s/r]$, and

  \item \label{lemma:fresh-variable:b} $\Gamma\vdash_R\varphi$ iff
    $\Gamma[s/r]\vdash_R\varphi[s/r]$.
  \end{enumerate}
\end{lemma}
\begin{proof}
  \ref{lemma:fresh-variable:a} follows by induction on the length of
  $\QRBB$ derivations.  \ref{lemma:fresh-variable:b} follows by
  induction on the length of $L(R)$-proofs and makes use of
  \ref{lemma:fresh-variable:a}.
\end{proof}

\begin{lemma}[Conservativity]
  \label{lemma:conservativity}
  Assume $R$ is at least countably infinite.  If $T$ is a theory in
  $L(R)$, then $T^*$ is a conservative extension of $T$.
\end{lemma}
\begin{proof} 
  We prove that for each $\varphi\in L(R)$, we have
  $T^*\vdash_{R^*}\varphi$ iff $T\vdash_R\varphi$.  The right-to-left
  direction is immediate, so we only need prove the left-to-right
  direction.  Proceeding, if $\varphi\in L(R)$, then we have
  $T^*\vdash_{R^*}\varphi$ iff there exists a finite set $H\subseteq
  H(R)$ of Henkin axioms satisfying $T,H\vdash_{R^*}\varphi$.  So it
  suffices for us to prove by induction on the finite cardinality of
  $H\subseteq H(R)$ that $T,H\vdash_{R^*}\varphi$ implies
  $T\vdash_R\varphi$.  The base case (where $H=\emptyset$) follows by
  Lemma~\ref{lemma:constants}, so we proceed directly to the induction
  step.  That is, we assume that the result holds for $H\subseteq
  H(R)$ having $|H|=n$ and we prove the result holds for $H\subseteq
  H(R)$ having $|H|=n+1$.  Proceeding, take $H\subseteq H(R)$
  satisfying $|H|=n+1$, choose a Henkin axiom $h\in H$ with
  \[
  h=\lnot(\forall r)\psi\to\lnot\psi[r_\psi/r]\enspace,
  \]
  and define $H':=H-\{h\}$ so that $H=H'\cup\{h\}$ and $|H'|=n$.  Now
  assume $T,H\vdash_{R^*}\varphi$ with $\varphi\in L(R)$, and hence
  $T,H',h\vdash_{R^*}\varphi$.  It follows that there is a finite
  $T'\subseteq T$ such that $T',H',h\vdash_{R^*}\varphi$.  Let
  $s\in R$ be a variable not occuring in any formula in the finite set
  $T'\cup H'\cup\{h,\varphi\}$.  Such $s$ exists because $R$ is at
  least countably infinite. Define
  \[
  h':=\lnot(\forall r)\psi\to\lnot\psi[s/r]\enspace.
  \]
  Then, omitting mention of instances of classical reasoning and the
  use of the Simple Lemma in the last six lines of the derivation, we
  have:
  \begin{center}
  \renewcommand{\arraystretch}{1.3}\begin{tabular}{lll}
    &
    $T',H',h\vdash_{R^*}\varphi$
    & (derived above)
    \\ &
    $T',H',h'\vdash_{R^*}\varphi$
    & Lemma~\ref{lemma:fresh-variable}
    \\ &
    $\vdash_{R^*}\bigwedge(T'\cup H')\to(h'\to\varphi)$
    & Lemma~\ref{lemma:deduction}
    \\ &
    $\vdash\bigwedge(T'\cup H')\to(h'\to\varphi)$
    & Simple Lemma
    \\ &
    $\vdash(\forall s)(\bigwedge(T'\cup H')\to(h'\to\varphi))$
    & (Gen)
    \\ &
    $\vdash\bigwedge(T'\cup H')\to(\forall s)(h'\to\varphi)$
    & (UD), no $s$ in $T'\cup H'$
    \\ &
    $\vdash_{R^*}\bigwedge(T'\cup H')\to(\forall s)(h'\to\varphi)$
    & Simple Lemma
    \\ &
    $T',H'\vdash_{R^*}(\forall s)(h'\to\varphi)$
    & Lemma~\ref{lemma:deduction}
    \\ &
    $T',H'\vdash_{R^*}
    (\forall s)((\lnot(\forall r)\psi\to\lnot\psi[s/r])
    \to\varphi)$
    & write out $h'$
    \\ &
    $T',H'\vdash_{R^*}
    (\lnot(\forall r)\psi\to(\exists s)\lnot\psi[s/r])
    \to\varphi$
    & Lem.~\ref{lemma:quantifiers}, no $s$  in $\varphi$ or $\lnot(\forall r)\psi$
    \\ &
    $T',H'\vdash_{R^*}
    (\lnot(\forall r)\psi\to\lnot(\forall s)\lnot\lnot\psi[s/r])
    \to\varphi$
    & definition of $\exists$
    \\ &
    $T',H'\vdash_{R^*}
    (\lnot(\forall r)\psi\to\lnot(\forall s)\psi[s/r])
    \to\varphi$
    & Equivalence (Lemma~\ref{lemma:rules})
    \\ &
    $T',H'\vdash_{R^*}\varphi$
    & Renaming (Lemma~\ref{lemma:rules})
  \end{tabular}
  \end{center}
  Therefore $T,H'\vdash_{R^*}\varphi$. Applying the induction
  hypothesis, $T\vdash_R\varphi$.
\end{proof}

Now, given a theory $T$ in $L(R)$, define:
\[
\begin{array}{lcl}
  T_0 &=& T \\
  T_{i+1} &=& (T_i)^* \text{ for } i\in\omega \\
  T_\omega &=& \textstyle\bigcup_{i\in\omega} T_i
\end{array}
\quad
\begin{array}{lcl}
  R_0 &=& R \\
  R_{i+1} &=& (R_i)^* \text{ for } i\in\omega \\
  R_\omega &=& \textstyle\bigcup_{i\in\omega} R_i
\end{array}
\]

\begin{lemma}[Henkin]
  \label{lemma:henkin}
  Let $R$ be at least countably infinite and $T$ be a theory in
  $L(R)$.  Then $T_\omega$ is a Henkin theory that is a conservative
  extension of $T$.
\end{lemma}
\begin{proof}
  Take a closed $\lnot(\forall r)\varphi\in L(R_\omega)$.  Then there
  exists $i\in\omega$ such that $\lnot(\forall r)\varphi\in L(R_i)$.
  But then there is a witness $r_\varphi\in L(R_{i+1})\subseteq
  L(R_\omega)$ such that
  \[
  \lnot(\forall r)\varphi\to\lnot\varphi[r_\varphi/r]\in
  T_{i+1}\subseteq T_\omega\enspace.
  \]
  So $T_\omega$ is a Henkin theory.

  By induction on $i\in\omega$, we prove that $T_i$ is a conservative
  extension of $T$.  Base case: $T_0=T$ and the result is immediate.
  Induction step: $T_{i+1}$ is a conservative extension of $T_i$ by
  Lemma~\ref{lemma:conservativity}; that is, $T_{i+1}\cap L(R_i)=T_i$.
  By the induction hypothesis, $T_i\cap L(R)=T$.  But then, since
  $L(R)\subseteq L(R_j)\subseteq L(R_k)$ if $j<k$, we have
  \[
  T_{i+1}\cap L(R)=
  (T_{i+1}\cap L(R_i))\cap L(R)=
  T_i\cap L(R)=
  T\enspace.
  \]
  So each $T_i$ is a conservative extension of $T$. But then
  \[
  \textstyle
  T_\omega\cap L(R)=
  (\bigcup_{i\in\omega}T_i)\cap L(R)=
  \bigcup_{i\in\omega}(T_i\cap L(R))=
  T\enspace,
  \]
  which shows that $T_\omega$ is a conservative extension of $T$.
\end{proof}

By the usual Lindenbaum argument (using Zorn's Lemma)
\cite[\S3.1]{vDal94:Book}, for each $R$, any consistent set in $L(R)$
may be extended to a maximal $L(R)$-consistent set.  Hence for a
consistent theory $T$ in $L(R)$, the theory $T_\omega$ in
$L(R_\omega)$ is consistent and may be extended to a maximal
$L(R_\omega)$-consistent set $T_\omega'$.  This set is a theory in
$L(R_\omega)$. Further, this theory is Henkin because
$T_\omega\subseteq T_\omega'$, both theories are in the same language,
$T_\omega$ is Henkin by Lemma~\ref{lemma:henkin}, and any extension of
a Henkin theory within the same language is still Henkin (because all
Henkin axioms are already present).

To prove completeness of $\QRBB$, we take $\theta$ such that
$\QRBB\nvdash\theta$. We construct a structure $M_c=(W,[\cdot],N,V)$
as in the proof of Theorem~\ref{theorem:RBB-determinacy} except that
our set of worlds $W$ is defined differently.  First, let $M_0$ be the
set of all maximal $L(R)$-consistent sets; each such set is a theory
in $L(R)$.  For each theory $T\in M_0$, define $M_\omega(T)$ to be the
set of all maximal $L(R_\omega)$-consistent extensions of $T_\omega$.
As we have seen, each member of $M_\omega(T)$ is a maximal
$L(R_\omega)$-consistent Henkin theory that is conservative over $T$
(Lemma~\ref{lemma:henkin}).  Define the set
\[
\textstyle M := \bigcup_{T\in M_0}M_\omega(T)
\]
whose members are maximal $L(R_\omega)$-consistent extensions of
$T_\omega$ for each $T\in M_0$. It follows that $\{\lnot\theta\}$ can
be extended to a $T^\theta\in M_0$ and hence neither
$M_\omega(T^\theta)$ nor $M$ is empty. We define $W:=M\times\{1,2\}$
and write $(\Gamma,i)\in W$ in the abbreviated form $\Gamma_i$.  Since
$M$ is nonempty, $W$ is nonempty.  The remaining components of $M_c$
are defined as in the proof of Theorem~\ref{theorem:RBB-determinacy}
except that all language-specific aspects of definitions are extended
to the larger language $L(R_\omega)$.

We prove the \emph{Truth Lemma\/}: for each formula
$\varphi\in L(R_\omega)$ and world $\Gamma_i\in W$, we have
$\varphi\in\Gamma$ iff $M_c,\Gamma_i\models\varphi$.  The proof is by
induction on the construction of formulas.  The arguments for all but
two cases are as in the proof of
Theorem~\ref{theorem:RBB-determinacy}.  All that remains are the
equality and quantifier inductive step cases.
\begin{itemize}
\item Inductive step: $(s=r)\in\Gamma$ iff $M_c,\Gamma_i\models s=r$.
  
  By (EP) and (EN), we have $(s=r)\in\Gamma$ iff $s=r$.  But the
  latter holds iff $M_c,\Gamma_i\models s=r$.

\item Inductive step: $(\forall r)\varphi\in\Gamma$
  iff $M_c,\Gamma_i\models(\forall r)\varphi$.

  If $(\forall r)\varphi\in\Gamma$, then it follows by maximal
  $L(R_\omega)$-consistency and (UI) that $\varphi[s/r]\in\Gamma$ for
  each $s\in R_\omega$ that is free for $r$ in $\varphi$.  By the
  induction hypothesis, we have $M_c,\Gamma_i\models\varphi[s/r]$ for
  each such $s\in R$.  But this is what it means to have
  $M_c,\Gamma_i\models(\forall r)\varphi$.
  
  Conversely, if $M_c,\Gamma_i\models(\forall r)\varphi$, then it
  follows that $M_c,\Gamma_i\models\varphi[s/r]$ for all $s\in
  R_\omega$ free for $r$ in $\varphi$.  By the induction hypothesis,
  we have $\varphi[s/r]\in\Gamma$ for all such $s$.  Since $\Gamma$ is
  a Henkin theory, there is a Henkin constant $r_\varphi\in R_\omega$
  for $\lnot(\forall r)\varphi$.  Hence
  $\varphi[r_\varphi/r]\in\Gamma$.  But $\Gamma$ contains the Henkin
  axiom
  \[
  \lnot(\forall r)\varphi\to\lnot\varphi[r_\varphi/r]\enspace,
  \]
  and so we have by $L(R_\omega)$-consistency that $(\forall
  r)\varphi\in\Gamma$.
\end{itemize}
This completes the proof of the Truth Lemma.

The proof that $M_c$ satisfies (pr), (brk), (ba), (as), and (d) is as
in the proof of Theorem~\ref{theorem:RBB-determinacy}. So $M_c$ is
indeed a model (and not just a pre-model).

To complete the proof of completeness, we recall that we obtained
$T^\theta\in M_0$ as a maximal $L(R)$-consistent extension of
$\{\lnot\theta\}$.  Hence there exists
$\Gamma^\theta\in M_\omega(T^\theta)\subseteq M$. But $\Gamma^\theta$
is a maximal $L(R_\omega)$-consistent extension of
$(T^\theta)_\omega$, and $(T^\theta)_\omega$ is a conservative
extension of $T^\theta$ by Lemma~\ref{lemma:henkin}.  Therefore, since
$\theta\notin T^\theta$ by consistency, we have
$\theta\notin\Gamma^\theta$.  Applying the Truth Lemma,
$M_c,\Gamma^\theta_1\not\models\theta$.  Completeness follows.

\subsection{Conservativity of $\QRBB$ Over $\RBB$}

As a corollary of Theorems~\ref{theorem:RBB-determinacy} and
\ref{theorem:QRBB-determinacy}, we have the following.

\begin{corollary}
  \label{corollary:extension}
  $\QRBB$ is a conservative extension of $\RBB$: for each $\varphi\in
  F$,
  \[
  \QRBB\vdash\varphi
  \quad\text{iff}\quad
  \RBB\vdash\varphi\enspace.
  \]
\end{corollary}
\begin{proof}
  The right-to-left direction is obvious ($\QRBB$ contains all the
  axioms and rules of $\RBB$).  The left-to-right direction follows by
  $\QRBB$ soundness (Theorem~\ref{theorem:QRBB-determinacy}) and
  $\RBB$ completeness (Theorem~\ref{theorem:RBB-determinacy}).
\end{proof}

\subsection{$\QRBB_\sigma$ and $\QRBB_\sigma^+$ Soundness and
  Completeness}

Recalling the semantics for $\QRBB_\sigma$ and for $\QRBB_\sigma^+$
from \S\ref{section:QRBB-semantics}, we prove the following theorem.

\begin{theorem}[$\QRBB_\sigma$ and $\QRBB_\sigma^+$ Soundness and
  Completeness]
  \label{theorem:QRBBsigma-determinacy}
  Assume $R$ contains the symbol $\sigma$.
  \begin{itemize}
  \item $\QRBB_\sigma$ is sound: $\QRBB_\sigma\vdash\varphi$ implies
    $\models_\sigma\varphi$ for each $\varphi\in F^\forall$.

  \item if $R$ is at least countably infinite, then $\QRBB_\sigma$ is
    sound and complete: for each $\varphi\in F^\forall$,
    \[
    \QRBB_\sigma\vdash\varphi
    \quad\text{iff}\quad\models_\sigma\varphi\enspace.
    \]

  \item analogous soundness and completeness results hold for
    $\QRBB_\sigma^+$ with respect to the satisfaction relation
    $\models_\sigma^+$.
  \end{itemize}
\end{theorem}

Soundness is proved as in Theorem~\ref{theorem:RBBsigma-determinacy}.
Completeness is proved as in Theorem~\ref{theorem:QRBB-determinacy},
except that provability is taken with respect to either $\QRBB_\sigma$
or $\QRBB_\sigma^+$ and one must show (using an argument as in the
completeness portion of Theorem~\ref{theorem:RBBsigma-determinacy})
that $M_c$ satisfies the relevant properties.


\bibliographystyle{plain}
\bibliography{emr-gettier}

\begin{thebibliography}{10}

\bibitem{Arm1973:belief}
David~M. Armstrong.
\newblock {\em Belief, truth and knowledge}.
\newblock CUP Archive, 1973.

\bibitem{Art08:RSL}
Sergei~[N.] Artemov.
\newblock The logic of justification.
\newblock 1(4):477--513, December 2008.

\bibitem{ArtFit11:SEP}
Sergei~[N.] Artemov and Melvin Fitting.
\newblock Justification logic.
\newblock In Edward~N. Zalta, editor, {\em The {S}tanford {E}ncyclopedia of
  {P}hilosophy}, 2011.

\bibitem{BalRenSme12:LNCS}
Alexandru Baltag, Bryan Renne, and Sonja Smets.
\newblock The logic of justified belief change, soft evidence and defeasible
  knowledge.
\newblock In L.~Ong and R.~de~Queiroz, editors, {\em Logic, Language,
  Information and Computation: 19th International Workshop, {WoLLIC~2012},
  Buenos Aires, Argentina, September 3--6, 2012. Proceedings}, volume 7456 of
  {\em Lecture Notes in Computer Science}, pages 168--190, Buenos Aires,
  Argentina, 2012. Springer-Verlag Berlin Heidelberg.

\bibitem{BalRenSme14:APAL}
Alexandru Baltag, Bryan Renne, and Sonja Smets.
\newblock The logic of justified belief, explicit knowledge, and conclusive
  evidence.
\newblock {\em Annals of Pure and Applied Logic}, 165(1):49--81, 2014.

\bibitem{Bra11:SEP}
Torben Bra{\"u}ner.
\newblock Hybrid logic.
\newblock In Edward~N. Zalta, editor, {\em The {S}tanford {E}ncyclopedia of
  {P}hilosophy}, 2011.

\bibitem{Chellas:ml}
Brian~F. Chellas.
\newblock {\em Modal Logic: An Introduction}.
\newblock Cambridge University Press, 1980.

\bibitem{Chis77:TOK}
Roderick~M. Chisholm.
\newblock {\em Theory of knowledge}.
\newblock Prentice Hall, Englewood Cliffs, NJ, 1977.
\newblock 2nd edition.

\bibitem{Cla63:A}
Michael Clark.
\newblock Knowledge and grounds: A comment on {M}r.~{G}ettier's paper.
\newblock {\em Analysis}, 24(2):46--48, 1963.

\bibitem{Dre71:CR}
Fred Dretske.
\newblock Conclusive reasons.
\newblock {\em Australasian Journal of Philosophy}, 1(49):1--22, 1971.

\bibitem{Dut10:PHD}
Julien Dutant.
\newblock {\em Knowledge, Methods, and the Impossibility of Error}.
\newblock 2010.
\newblock PhD, University of Geneva.

\bibitem{Dut15:Method}
Julien Dutant.
\newblock Methods models for belief and knowledge.
\newblock 2015.
\newblock Manuscript, London.

\bibitem{egre2017knowledge}
Paul \'Egr{\'e}.
\newblock Knowledge as \emph{de re} true belief?
\newblock {\em Synthese}, 194(5):1517--1529, 2017.

\bibitem{Fit10:strengthening}
Branden Fitelson.
\newblock Strengthening the case for knowledge from falsehood.
\newblock {\em Analysis}, 70(4):666--669, 2010.

\bibitem{Fum02:TOJ}
Richard Fumerton.
\newblock Theories of justification.
\newblock {\em The Oxford handbook of epistemology}, pages 204--233, 2002.

\bibitem{Get63:A}
Edmund Gettier.
\newblock Is justified true belief knowledge?
\newblock {\em Analysis}, 23:121--123, 1963.

\bibitem{Gol76:JP}
Alvin~I. Goldman.
\newblock Discrimination and perceptual knowledge.
\newblock {\em The Journal of Philosophy}, 73(20):771--791, 1976.

\bibitem{Gol79:WJB}
Alvin~I. Goldman.
\newblock What is justified belief?
\newblock In G.~S. Pappas, editor, {\em Justification and Knowledge}, pages
  1--23. Dordrecht: Reidel, 1979.

\bibitem{Har70:KRC}
Gilbert~H. Harman.
\newblock Knowledge, reasons, and causes.
\newblock {\em The Journal of Philosophy}, pages 841--855, 1970.

\bibitem{kahneman2002representativeness}
Daniel Kahneman and Shane Frederick.
\newblock Representativeness revisited: Attribute substitution in intuitive
  judgment.
\newblock In T.~Gilovich, D.~Griffin, and D.~Kahneman, editors, {\em Heuristics
  and biases: The psychology of intuitive judgment}, pages 49--81. Cambridge
  University Press, 2002.

\bibitem{Kra02:LP}
Angelika Kratzer.
\newblock Fact: Particulars or information units?
\newblock {\em Linguistics and Philosophy}, 25:655--670, 2002.

\bibitem{LehPax69:JP}
K.~Lehrer and T.~Jr Paxson.
\newblock Knowledge: Undefeated justified true belief.
\newblock {\em Journal of Philosophy}, 66:225--237, 1969.

\bibitem{Leh00:Book}
Keith Lehrer.
\newblock {\em Theory of Knowledge}.
\newblock Westview Press, 2nd edition, 2000.

\bibitem{Neta11:refutation}
Ram Neta.
\newblock A refutation of {C}artesian fallibilism.
\newblock {\em No{\^u}s}, 45(4):658--695, 2011.

\bibitem{Sor15:fugu}
Roy Sorensen.
\newblock Fugu for logicians.
\newblock {\em Philosophy and Phenomenological Research}, 2015.
\newblock Forthcoming.

\bibitem{Sos74:how}
Ernest Sosa.
\newblock How do you know?
\newblock {\em American Philosophical Quarterly}, pages 113--122, 1974.

\bibitem{Sos79:presupp}
Ernest Sosa.
\newblock Epistemic presupposition.
\newblock In G.~Pappas, editor, {\em Justification and Knowledge}, pages
  79--92. Springer, 1979.

\bibitem{spinoza1677ethica}
Spinoza.
\newblock {\em Ethica}.
\newblock 1677.
\newblock Latin and English versions available at \url{http://ethicadb.org/}.

\bibitem{Tur12:S}
John Turri.
\newblock Is knowledge justified true belief?
\newblock {\em Synthese}, 184(3):247--259, 2012.

\bibitem{BenPac11:evidence}
Johan van Benthem and Eric Pacuit.
\newblock Dynamic logics of evidence-based beliefs.
\newblock {\em Studia Logica}, 99(1-3):61--92, 2011.

\bibitem{vDal94:Book}
Dirk van Dalen.
\newblock {\em Logic and Structure}.
\newblock Springer-Verlag, 1994.

\bibitem{War05:KFF}
T.~Warfield.
\newblock Knowledge from falsehood.
\newblock {\em Philosophical Perspectives}, pages 405--416, 2005.

\bibitem{Wil00:Book}
Timothy Williamson.
\newblock {\em Knowledge and its Limits}.
\newblock Oxford University Press, 2000.

\bibitem{Zag94:PQ}
Linda Zagzebski.
\newblock The inescapability of {G}ettier problems.
\newblock {\em The Philosophical Quarterly}, 44(74):65--73, 1994.

\end{thebibliography}

\end{document}